\documentclass{article}
\usepackage{mathrsfs}
\usepackage{hyperref}
\usepackage{amssymb}
\usepackage{amsthm}
\usepackage{amsmath}
\usepackage{bbold}
\usepackage{color}

\newtheorem{theorem}{Theorem }[section]
\newtheorem{proposition}[theorem]{Proposition}

\newtheorem{lemma}[theorem]{Lemma}

\newtheorem{corollary}[theorem]{Corollary}
\newtheorem{remark}[theorem]{Remark}

\begin{document}

\title{{The relativistic KMS condition for the thermal $n$-point functions of the $P(\phi)_2$ model}}
%\titlerunning{The relativistic KMS condition}

\author{Christian D.\ J\"akel\protect\footnote{christian.jaekel@mac.com}   \and Florian Robl\protect\footnote{roblf1@cf.ac.uk}
\and {\small School of  Mathematics, Cardiff University, Wales, UK}}
%\institute{School of  Mathematics\\ Cardiff University, Wales, United Kingdom\\ 
%\email{christian.jaekel@mac.com} \\
%\email{roblf1@cf.ac.uk}}
%\authorrunning{C. J\"akel and F. Robl}

%\date{Submitted: March 17,  2011 / Accepted: not yet}
%\communicated{M. Salmhofer}

\maketitle
\begin{abstract}
Thermal quantum field theories are expected to obey a {\em relativistic KMS condition}, 
which replaces both the relativistic spectrum condition of Wightman quantum field theory 
and the KMS condition characterising equilibrium states in quantum statistical mechanics.

In a previous work it has been shown that the two-point function of the thermal ${\mathscr P}(\varphi)_2$ 
model satisfies the relativistic KMS condition. Here we extend this result to general $n$-point functions. 
In addition, we verify that the thermal Wightman distributions are tempered.
\end{abstract}

\section{Introduction}

For many practical purposes it may be sufficient to study thermal field theory in  finite spatial 
volume, where the Hamiltonian $H_{box}$ has discrete spectrum and is bounded from below. 
The thermal expectation value of an observable $O$ at temperature $\beta^{-1}$ is then given 
by the Gibbs state 
	\begin{equation}
		\label{gibbs}
		\langle O \rangle_\beta 
		:= \frac{ {\rm Tr } \, {\rm e}^{- \beta H_{box} } O} { {\rm Tr } \,  {\rm e}^{- \beta H_{box}}} \; .   
	\end{equation}
However,  if one wants to investigate  the structural properties resulting from Poincar\'e invariance 
of the underlying equations of motion, thermal equilibrium states  in infinite volume have to be considered.   
Fortunately, the appropriate generalisation of (\ref{gibbs}) to infinite volume is well-known: in finite volume 
the Gibbs states are characterised by their analyticity properties. The latter are summarised in the KMS 
condition (see, \emph{e.g.},  \cite{BR}). Haag, Hugenholtz and Winnink have shown that this 
characterisation remains valid in the thermodynamic limit~\cite{HHW}. In fact, in infinite volume one can 
{\em derive} the KMS condition from first principles, which characterise thermal equilibrium states 
phenomenologically. For example, one can derive the KMS condition from {\em passivity}~\cite{PW} 
or {\em stability} under small adiabatic perturbations of the dynamics~\cite{HKT-P}\cite{NT}.

Analyticity properties of the correlation functions were previously used by Wightman to characterise  
the vacuum state of a relativistic quantum field theory \cite{SW}. The  intention was to guarantee  
{\em stability} of the vacuum, and indeed the requested analyticity properties of the correlation functions ensure 
that the energy is bounded from below. But soon it turned out that also most peculiar {\em structural properties} 
(\emph{e.g.}, the Reeh-Schlieder property \cite{RSch}) follow from the analyticity properties of the correlation functions. 
These properties, first rejected as mathematical artifacts resulting from over-idealisation,  are nowadays 
considered to be characteristic for any proper quantum field theory. Experiments testing these properties 
have yet to be devised, but similar phenomena related to entanglement in quantum mechanics are intensely investigated. 

Although  vacuum states and thermal equilibrium states are both characterised by analyticity properties of the 
correlation functions, there is a pronounced difference between them: a state in thermal equilibrium cannot be 
invariant under Lorentz boosts \cite{Na}\cite{Oj}, even if the equations of motion are invariant under Poincar\'e
transformations and the propagation speed of signals is finite. Hitherto the correlation functions of a thermal 
state were required to be analytic only with respect to the time-direction distinguished by the rest-frame of the
equilibrium state. Structural results, which are similar to the ones derived from the spectrum condition of Wightman 
field theory, cannot be derived from the traditional KMS condition alone.

The picture changed fundamentally, when Bros and Buchholz~\cite{BB1} (see also \cite{BB2}\cite{Ku1}\cite{Ku2}) 
recognised that the passivity properties of an equilibrium state should be visible even  to an observer, who is 
moving with respect to the rest frame distinguished by the KMS state. Carefully evaluating the consequences, 
Bros and Buchholz suggested that 
the thermal correlation functions
of a relativistic system can be continued analytically  
into the tube domain $\mathbb{R}^{d} 
+
i \left(V^{+} \cap (\beta e - V^{+})\right)$, where~$\beta$ plays the role of the reciprocal temperature 
of the system, $e$ is the unit vector in the  time direction distinguished by the rest-frame, 
$d=2, 3, \ldots $ is the dimension of space-time
and $V^{+} := \{ (t, \vec x) 
\in \mathbb{R}^{d}  
\mid  | \vec x | < t  \} $ denotes the (open) forward light-cone. The consequences of the 
new {\em relativistic KMS condition}  are profound, aligning thermal field theory and vacuum 
field theory w.r.t.~basic, structural aspects~\cite{J}\cite{Jii}. 

The relativistic KMS condition has been established (see \cite{BB1}) for a large class of KMS states constructed 
by Buchholz and Junglas \cite{BJ}. Moreover, C.~G\'erard and the first author have shown that the relativistic 
KMS condition holds for the two-point function of the ${\mathscr P}(\varphi)_{2}$ model~\cite{GeJ3}. The 
present work extends the latter result to general $n$-point functions.

\paragraph{Content} In Section \ref{Sec2} we recall the Euclidean field theory on the cylinder. 
Using Minlos' theorem, we define Gaussian measures  on a space of distributions, supported on a cylinder. 
Following Glimm and Jaffe, we renormalise the interaction (Theorem \ref{TH1}) by  normal ordering the 
random variables.  This allows us to define the Euclidean ${\mathscr P}(\phi)_{2}$ model on the cylinder 
with a spatial cut-off $l \in \mathbb{R}^+$ (see also \cite{GeJ1}). The corresponding probability measure 
${\rm d} \mu_{l}$ is absolutely continuous with respect to the Gaussian measure. {\em Nelson symmetry} 
(see Theorem \ref{3.2} $ii.)$) can be used to remove the spatial cut-off: the Schwinger functions of the 
thermal ${\mathscr P}(\varphi)_2$-model on the two-dimensional Minkowski space at temperature~$\beta^{-1}$ 
exist and are equal  to the Schwinger functions of the vacuum  ${\mathscr P}(\varphi)_2$-model on the 
Einstein universe of spatial circumference $\beta$ (up to interchanging the interpretation of the Euclidean 
variables $(\alpha, x) \in S_\beta \times \mathbb{R}$). 

The Osterwalder-Schrader reconstructions, presented in  Subsection \ref{sectosrec},  provide 
\begin{itemize}
	\item[$ i.)$]
		a {\em thermal field theory on Minkowski space $\mathbb{R}^{1+1}$}, consisting of 
		\begin{itemize}
		\item[---] a Hilbert space ${\mathcal H}_\beta$, together with a distinguished 
			vector $\Omega_\beta \in {\mathcal H}_\beta$; 
		\item[---] a von Neumann algebra ${\mathcal R}_\beta\subset {\mathscr B} ({\mathcal H}_\beta)$, 
			together with an abelian subalgebra, generated by the bounded functions of the time-zero fields; 
		\item[---] a one-parameter group of time-translation automorphisms $\{ \tau_{t,0} \mid t \in \mathbb{R} \}$ 
				induced by unitary operators ${\rm e}^{i tL}$, with spectrum of $L $ equal to $ \mathbb{R}$; 
		\end{itemize}
	\item[$ ii.)$]
		a {\em vacuum theory on the Einstein universe} (a cylinder, with the position variable taking values on the circle $S_{\beta}$
		and the time variable real valued), consisting of 
		\begin{itemize}
			\item[---] a Hilbert space ${\mathcal H}_C$, together with a distinguished vector $\Omega_C \in {\mathcal H}_C$; 
			\item[---] a von Neumann algebra ${\mathcal R}_C\subset {\mathscr B} ({\mathcal H}_C)$, together with 
				an abelian subalgebra, generated by the bounded functions of the time-zero fields; 
		\item[---] a one-parameter group of time-translation automorphisms $\{ \tau'_{0,s} \mid s \in \mathbb{R} \}$ 
				induced by unitary operators ${\rm e}^{i sH_C}$, with $H_C \ge0$.
		\end{itemize}
\end{itemize}

In Subsection \ref{SSec2.4}  Wightman field theory on the circle is discussed in some more detail. We identify 
${\mathcal H}_C$ with the Fock space $\Gamma \bigl(H^{-\frac{1}{2}}(S_{\beta}) \bigr)$ over
the Sobolev space~$H^{-\frac{1}{2}}(S_{\beta})$ on~the circle $S_{\beta}$, and recall the Glimm-Jaffe $\phi$-bounds. 
According to Theorem \ref{HO}, due to Heifets and Osipov, the joint spectrum of $P_C$ and $H_C$ is contained 
in the forward light cone $
\widetilde{V}^+ := \{ (p, E) \mid | p | < E  \} $. 
Consequently  (Theorem \ref{5.1}) the Fourier 
transform of the Wightman $n$-point function  (expressed in relative variables) has support in~$(V^+ )^{n-1}$  and the 
Wightman $n$-point-distribution  ${{\mathfrak W}}_{C}^{(n-1)} $ itself  is the boundary value of a polynomially bounded
function ${{\mathcal W}}_{C}^{(n-1)} $, which is analytic in the forward tube $(S_{\beta}\times\mathbb{R}-iV^{+})^{n-1}$. 

Lemma \ref{5.2} characterises a set of space-time points in $S_\beta \times \mathbb{R}$, which are mutually space-like to each
other. Locality implies that the Wightman $n$-point functions are real valued, if evaluated at these points. Thus, using Schwarz's 
reflection principle, we can extend the Wightman $n$-point functions on the circle to functions, which (expressed in relative variables) 
are holomorphic  in 
	\begin{equation} 
		\label{c-n-1-1}
			{\mathcal D}^{(n-1)} := 
			\left( \lambda_1 V_{\beta} \times \ldots \times \lambda_{n-1} V_{\beta}\right)\, + \,  i \, (V^{+} \cup V^{-}) \times \ldots \times
			(V^{+} \cup V^{-}) \, , 
	\end{equation}
where $V_{\beta}:=\bigl\{(\alpha, s) \mid |s| < \alpha < \beta- |s| \bigr\}$, $\lambda_i >0$ and $\sum_{i=1}^{n-1} \lambda_i =1$. 
The Edge-of-the-Wedge theorem \cite[Theorem 2-16]{SW}\cite{BEGS} implies  that 
the tempered distributions ${{\mathfrak W}}_{C}^{(n-1)} $  are the boundary values of  functions  defined and holomorphic  in 
	\begin{equation}
		\label{intro5}
		{\mathcal C}^{(n-1)}  := {\mathcal N} \, \cup \, {\mathcal D}^{(n-1)}, 
	\end{equation}
where ${\mathcal N}$ is a complex neighbourhood of $\lambda_1 V_{\beta} \times \ldots \times \lambda_{n-1} V_{\beta}$ (Theorem \ref{Th5}).  

In Section 3 we return to the thermal ${\mathscr P}(\varphi)_2$ model on the two dimensional Minkowski space. 
Invoking~(\ref{c-n-1-1}), Nelson symmetry implies that 
analytic continuations 
of the thermal Wightman distributions ${{\mathfrak W}}_{\beta}^{(n-1)} $ can a priori
be defined (as analytic functions)  in the domain 
	\[
		(Q^- \cup Q^+)^{n-1} - i \left( \lambda_1 V_{\beta} \times \ldots \times \lambda_{n-1} V_{\beta}\right) \, , 
	\]
where the right and left space-like wedges are 
$Q^\pm = \left\{ (t, x) \in \mathbb{R}^2 \mid \pm x > |t| \right\}$, and, as before, $\lambda_i >0$ and $\sum_{i=1}^{n-1} \lambda_i =1$. 

The existence of products of thermal sharp-time fields is shown in Lemma~\ref{pos}. 
Taking advantage of their Euclidean heritage, their domain properties are summarised in  Proposition~\ref{neu}.
The  spectral theorem is used to extend the functions ${{\mathcal W}}_{\beta}^{(n-1)} $
to  functions holomorphic in 
the product of domains 
	\[
		(\lambda_1 {\cal T}_{\beta}) \times \cdots \times (\lambda_{n-1} {\cal T}_{\beta}), 
		\qquad  {\cal T}_{\beta} := \mathbb{R}^2 - i V_{\beta} , \qquad \sum_{j=1}^{n-1} \lambda_j = 1,
	\]
and $\lambda_j >0$, $j=1, \ldots, n-1$ (Theorem \ref{wightanal}). 
The final subsection  deals with the boundary values of these functions.
A generalisation of Ruelle's H\"older inequality for Gibbs states, suggested by J.~Fr\"ohlich  in~\cite{Fr1}, 
is presented  in Theorem \ref{hoelder}. A fractional $\phi$-bound, established in Proposition~\ref{prop2}, provides a 
key ingredient in the proof of the final result (Theorem \ref{Th7}), which establishes  that the thermal Wightman 
$n$-point-distributions  ${{\mathfrak W}}_{\beta}^{(n-1)} $ 
of the~${\mathscr P}(\varphi)_{2}$ model on the real line are {\em tempered distributions} which
satisfy the {\em relativistic KMS condition}.

\section{Euclidean fields on the cylinder}
\label{Sec2}

In 1974 H\o egh-Krohn  \cite{H-K} discovered that the Euclidean field theory
on the cylinder allows to reconstruct {\em two} distinct quantum field
theories. In this section we recall the main steps of the two reconstructions \cite{GeJ1,GeJ2}, leading to a vacuum theory 
on the Einstein universe (a cylinder, with the position variable taking values on a circle and the time variable real valued)
and a thermal theory on $1+1$-dimensional Minkowski space.

\subsection{Probability measures on the cylinder}

Consider a cylinder $S_{\beta}\times \mathbb{R}$, with $S_{\beta}$ the circle of circumference~$\beta$. The coordinates
$(\alpha,x) \in S_{\beta}\times \mathbb{R}$ of a point in the cylinder will refer to either one of the charts $[-\beta/2, \beta/2) \times \mathbb{R}$ 
or $[0, \beta) \times \mathbb{R}$. 

Let ${\mathscr S}(\mathbb{R})$ denote the set of Schwartz functions on the real line.
For consistency we denote the set of $C^\infty$-functions on the circle $S_{\beta}$ by
${\mathscr S}(S_{\beta})$.
The Fr\'echet space ${\mathscr S}(S_{\beta}\times \mathbb{R})$  is the space of smooth
functions $f$ on the cylinder, which are $\beta$-periodic in
$\alpha$ and fulfil 
	\[
		\left| \left(1+|x| \right)^{ k }\partial_{\alpha}^{p}\partial^{k}_{x}f(\alpha,x) \right| \leq C_{p, k} \,  ,
		\qquad p\in \mathbb{N}, \quad k \in \mathbb{N}. 
	\]

${\mathscr S}'(\mathbb{R})$, ${\mathscr S}'(S_{\beta})$ and ${\mathscr S}'(S_{\beta}\times \mathbb{R})$ denote the  dual spaces of
${\mathscr S}(\mathbb{R})$, ${\mathscr S}(S_{\beta})$ and ${\mathscr S}(S_{\beta}\times\mathbb{R})$. The real-linear subspaces of real
valued distributions are indicated by ${\mathscr S}'_{\mathbb{R}}(\mathbb{R})$, ${\mathscr S}'_{\mathbb{R}}(S_{\beta})$ and
$Q:={\mathscr S}'_{\mathbb{R}}(S_{\beta}\times \mathbb{R})$. The Borel $\sigma$-algebra $\Sigma$ on $Q$  is the minimal 
$\sigma$-algebra containing all open sets in the $\sigma({\mathscr S}', {\mathscr S})$-topology. The evaluation 
map~$\phi(f)$, $f\in {\mathscr S}_{\mathbb{R}}(S_{\beta}\times\mathbb{R})$, 
	\[
		\phi(f)\colon Q  \to  \mathbb{R}, \quad q \mapsto  \langle q, f\rangle,
	\]
is defined in terms of the duality bracket $\langle  \, . \, , \, .\, \rangle$. In the present context $\phi$ is called the {\em Euclidean quantum field}. 

If~${\rm d} \mu$ is a probability measure  on the space $(Q, \Sigma)$, then its Fourier transform 
	\[
		E(f) = \int_{Q} {\rm e}^{i \phi(f)} {\rm d} \mu ,  \quad f\in {\mathscr S}_{\mathbb{R}}(S_{\beta}\times \mathbb{R}),
	\]
satisfies 
\begin{itemize}
	\item [$ i.)$] $E(0) = 1$; 
	\item [$ ii.)$] $ {\mathscr S}_{\mathbb{R}}(S_{\beta}\times \mathbb{R}) \ni f \mapsto E(f) \in \mathbb{C} $ is continuous; 
	\item [$ iii.)$] for all $f_i, f_j \in {\mathscr S}_{\mathbb{R}}(S_{\beta}\times \mathbb{R})$ and $z_i, z_j \in \mathbb{C}$,
			$i,j = 1, \ldots, n$, 
			\[
				\sum_{i,j =1}^n z_i \bar z_j E (f_i - f_j) \ge 0 . 
			\]
\end{itemize}
\smallskip
On the converse, {\em Minlos' theorem} \cite{Bour}\cite{GJ}\cite{M}\cite{Si1} 
states that any function $E$ on ${\mathscr S}_{\mathbb{R}}(S_{\beta}\times \mathbb{R})$ satisfying 
the properties~$i.)$--$iii.)$ is the Fourier transform of a probability measure~$ {\rm d} \mu$ on Q. 

Generating functionals of the form
	\begin{equation}
		\label{e1.00}
			E_0(f) = {\rm e}^{- C(f,f)/2}, \quad f\in {\mathscr S}_{\mathbb{R}}(S_{\beta}\times \mathbb{R}) , 
	\end{equation}
with $C(\, . \, , \, . \,)$ a weakly continuous positive semi-definite quadratic form, clearly satisfy the conditions $i.)$--$iii.)$ of 
Minlos' theorem and thus  give rise to probability measures on~$(Q, \Sigma)$.  These measures are called {\em Gaussian} 
measures. The Gaussian measure  on $Q$, with  {\em covariance}
	\begin{equation}
		\label{cov}
		C(f_1, f_2):= \left(f_1, (D^2_\alpha + D^2_ x + m^2)^{-1}f_2 \right), 
		\quad f_1,f_2\in {\mathscr S} (S_{\beta}\times\mathbb{R})  ,
	\end{equation} 
is denoted by ${\rm d}\phi_{C}$. The  scalar product $(\, .\, , \, . \, )$ in (\ref{cov}) refers 
to~$L^{2}(S_{\beta}\times \mathbb{R})$ and $D_{\alpha}= - i \partial_{\alpha}$, $D_{x}= -i \partial_{x}$. 
The Euclidean quantum field $\phi$ on the cylinder is called {\em free}, if $\phi(f)$ is 
viewed as a measurable function on the probability space $(Q,\Sigma, {\rm d}\phi_{C})$.

In this work we are interested in non-Gaussian measures (see Theorem \ref{3.2} below). 
They formally result from adding a polynomial of the form~${\mathscr P}(\varphi)$, where 
${\mathscr P}(\lambda)$, $\lambda \in \mathbb{R}$, is a  polynomial which is bounded from below, to the Hamiltonian 
of the free massive boson field.

In two dimensions, the singularities, which arise from taking powers of the Euclidean field $\phi$ at a 
point $(\alpha,x) \in S_{\beta}\times \mathbb{R}$, can be removed  by first normal ordering~$:\! . \!:_{c}$  (see \cite{GJ,Si1}) 
the monomials~$\phi(f)^{n}$, $n \in \mathbb{N}$,  
	\begin{equation}
		\label{wick}
		:\!\phi(f)^{n}\!:_{c} \; \;  := \sum_{m=0}^{[n/2]}\frac{n!}{m!(n-2m)!} \; \phi(f)^{n-2m} \Bigl(-\frac{1}{2}  c(f,f) \Bigr)^{m} 
	\end{equation}
with respect to a covariance~$c$, and then taking appropriate limits. $[\, . \, ]$ denotes taking the integer part. We will normal 
order with respect to different covariances~$c$, some of them being limiting cases of the covariance $C$ defined in (\ref{cov}).

Normal-ordering of point-like fields is  ill-defined (\emph{i.e.}, one cannot replace the test function~$f\in
{\mathscr S}_{\mathbb{R}}(S_{\beta}\times \mathbb{R})$ in (\ref{wick}) by a two dimensional Dirac $\delta$-function), 
but integrals over normal-ordered point-like fields can be defined rigorously: set, for~$k\in \mathbb{N}$ and 
$\kappa \in \mathbb{R}^+$,
	\[
		\delta_{k}(\alpha):=\beta^{-1}\sum_{|n|\leq k}{\rm e}^{i \nu_n \alpha } \quad \hbox{and} \quad
		\delta_{\kappa}(x):= \kappa \chi(\kappa x),
	\] 
where $\nu_n = 2\pi n/\beta$, $n \in \mathbb{N}$, and $\chi$ is an arbitrary, absolutely integrable function in 
$C_0^\infty(\mathbb{R})$ with $\int \chi(x)\, {\rm d} x=1$. With these notations we have the following result due to Glimm and Jaffe:

\begin{theorem}[Ultraviolet renormalisation \cite{GeJ2}\cite{GJ}]  
\label{TH1}
For $f\in L^{1}(S_{\beta} \times \mathbb{R} )\cap L^{2}(S_{\beta} \times \mathbb{R})$, the following  limit exists  
in $\kern -.2cm \bigcap \limits_{1\leq p<\infty} \kern -.2cm L^{p}(Q, \Sigma, {\rm d}\phi_{C})$:
	\begin{equation}
		\label{uvren}
		\: \lim_{k, \kappa\to \infty}\int_{S_{\beta}\times \mathbb{R}}f(\alpha,x):\!
		\phi \bigl(\delta_{k}(.-\alpha)\otimes \delta_{\kappa}(.-x) \bigr)^{n}\!:_{C}{\rm d} \alpha \, {\rm d} x .
	\end{equation}
We denote it by $ \int_{S_{\beta}\times \mathbb{R}}f(\alpha,x):\! \phi(\alpha,x)^{n}\!:_{C}{\rm d} \alpha {\rm d} x$. 
\end{theorem}

\paragraph{Remark} This key theorem, which follows from exactly the same arguments as in the vacuum case analyzed by 
Glimm and Jaffe \cite{GJ},  establishes a crucial step forward in the construction of the ${\mathscr P}(\varphi)_2$ model in finite volume, 
as it takes care (see Eq.~(\ref{co}) below) of the {\em ultraviolet renormalisation}. 

\bigskip
Let $ {\mathscr P} (\lambda) = \sum_{j}c_j \lambda^j$ be a real valued polynomial, which
is bounded from below. Replacing the function $f$ in (\ref{uvren}) by the characteristic function 
of the set $S_\beta \times [-l , l]$, $l \in \mathbb{R}^+$, and applying \cite[Lemma V.5]{Si1}, we deduce that 
	\begin{equation}
		\label{co}
		{\rm e}^{-\int_{-\beta/2}^{\beta/2}\int_{-l}^{l}  {:} {\mathscr P} (\phi (\alpha,x)){:}_C \; {\rm d} \alpha {\rm d} x}
		\in L^1 (Q, \Sigma, {\rm d}\phi_{C})  \quad \hbox{if} \quad 0< l < \infty \; . 
	\end{equation}
The {\em Euclidean ${\mathscr P}(\phi)_{2}$ model} on the cylinder with a spatial cut-off $l \in \mathbb{R}^+$ is specified by setting
	\begin{equation} 
		\label{cutmeasure}
		{\rm d}\mu_{l}:=\frac{1}{Z_l} \, {\rm e}^{-\int_{-\beta/2}^{\beta/2}\int_{-l}^{l}   {:} {\mathscr P} (\phi (\alpha,x)){:}_C \; 
		{\rm d} \alpha {\rm d} x} {\rm d}\phi_{C}. 
	\end{equation} 
The partition function $ Z_l$ is chosen such that $\int_{Q} {\rm d}\mu_l =1$. If $l < \infty$, then the measure ${\rm d} \mu_l$
is absolutely continuous with respect to the Gaussian measure ${\rm d}\phi_{C}$, with Radon-Nikodym derivative
${\rm d} \mu_l / {\rm d}\phi_{C}$ given by (\ref{co}). However, the limit of the functions 
in (\ref{co}) fails to be in $L^1(Q,\Sigma,{\rm d}\phi_{C})$ as $l \to \infty$, and therefore
the formal limiting measure can  not be absolutely continuous with respect to the Gaussian measure. 
In fact, in order to show that  a countably additive Borel measure exists in the limit $l \to \infty$, 
it is sufficient to show (see Theorem~\ref{3.2} below) that
	\begin{equation}
		\lim_{l\to +\infty}\int_{Q}
		{\rm e}^{ i \phi(f)}{\rm d}\mu_{l} = :  E_{\mathscr P}(f) , \qquad   f\in {\mathscr S} (S_{\beta}\times \mathbb{R}), 
	\label{limit}
	\end{equation}
defines a generating functional on ${\mathscr S}_{\mathbb{R}}'(S_{\beta}\times\mathbb{R})$ satisfying 
the properties~$i.)$--$iii.)$ of Minlos' theorem. 

\subsection{Sharp-time fields, Existence of the Euclidean measure in the thermodynamic limit, and Nelson symmetry}
\label{sectnelsym}

Cluster expansions (see \emph{e.g.}~\cite{GJ}) certainly allow one to control the limit in (\ref{limit}). But for the  
thermal ${\mathscr P}(\varphi)_2$ model, in which we are interested, there is another option, which was first explored in this context 
by H\o egh-Krohn \cite{H-K}: {\em Nelson symmetry}. It results from replacing the product measure ${\rm d} \alpha {\rm d} x$ in the exponent 
in (\ref{cutmeasure}) by iterated integrals with respect to the two measures ${\rm d} \alpha$ and $ {\rm d} x$,   in different orders. 
The delicate point, which will now be addressed in some more detail, is that {\em one} of the limits in~(\ref{uvren}) can be interchanged 
with {\em one} of the integrations. 

\goodbreak
In \cite {GeJ2} it has been shown that
	\begin{itemize}
		\item[$ i.)$] for $h_{1}, h_{2} \in {\mathscr S}_{\mathbb{R}}(\mathbb{R})$ and $0\leq \alpha_{1} , \alpha_2 \leq \beta$, 
			\begin{align}
				\label{gre}
				& \lim_{k,k'\to \infty} C \bigl(\delta_{k}(.-\alpha_{1})\otimes h_{1},
					\delta_{k'}(.-\alpha_{2})\otimes h_{2} \bigr) 
				\nonumber\\  
				& \qquad \qquad 
			= \Bigl( h_{1}, \frac{{\rm e}^{-|\alpha_{1}-
				\alpha_{2}|\epsilon}+ {\rm e}^{-(\beta-|\alpha_{1}-\alpha_{2}|)\epsilon}}{
				2\epsilon(1-{\rm e}^{- \beta \epsilon})}h_{2} \Bigr)_{L^{2}(\mathbb{R}, {\rm d} x)},
			\end{align}
			with $\epsilon:= \left(D_{x}^{2}+ m^{2}\right)^{\frac{1}{2} }$;
		\item[$ ii.)$] for $g_{1} , g_{2} \in {\mathscr S}_{\mathbb{R}}(S_{\beta})$ and $x_1, x_2 \in \mathbb{R}$,  
			\begin{equation}
				\label{gref}
				\lim_{\kappa , \kappa' \to \infty}C\bigl(g_{1}\otimes \delta_{\kappa}(.-x_{1}), g_{2}\otimes
				\delta_{\kappa'}(.-x_{2}) \bigr)= \Bigl(g_{1}, \frac{{\rm e}^{- |x_{1}- x_{2}| \nu}}{
				2\nu}g_{2} \Bigr)_{L^{2}( S_{\beta}, {\rm d} \alpha)},
			\end{equation}
		with $\nu:= \left(D_{\alpha}^{2}+ m^{2}\right)^{\frac{1}{2} }$. 
\end{itemize}

\smallskip
\noindent
Thus, for $h\in {\mathscr S}_{\mathbb{R}}(\mathbb{R})$,  $g\in {\mathscr S}_{\mathbb{R}}(S_{\beta})$ and
$\alpha\in S_{\beta}$, $x\in \mathbb{R}$ fixed, the sequences of functions
	\[ 
		\left\{ \phi \bigl( \delta_{k}(.-\alpha)\otimes h\bigr) \right\}_{k \in \mathbb{N}} \quad
		\hbox{and} \quad \left\{ \phi \bigl(g\otimes\delta_{\kappa}( .-x) \bigr) \right\}_{\kappa \in \mathbb{N}}
	\]
are Cauchy sequences in $\bigcap_{1\leq p<\infty}L^{p}(Q, \Sigma, {\rm d} \phi_{C})$. This can be derived from
the definition of the generating Gaussian functional, as (\ref{e1.00}) implies
	\begin{equation}
		\int_{Q}\phi(f)^{p}{\rm d}\phi_{C}= 
		\begin{cases}
		0, & p \text{ odd} \, ,\\
		(p-1)!! \; C(f,f)^{p/2},& p \text{ even} \, ,
		\end{cases} 
		\label{e1.0}
	\end{equation}
with $n!!= n(n-2)(n-4)\cdots 1$.  We can therefore define sharp-time fields
	\begin{equation}
		\phi (\alpha, h):=\lim_{k\to \infty}\phi \bigl(\delta_{k}(.-\alpha)\otimes h \bigr), 
		\quad  \phi (g, x):= \lim_{\kappa\to \infty}\phi \bigl(g\otimes\delta_{\kappa}(.-x) \bigr). 
		\label{e1.4}
	\end{equation}
We note that  both $\phi (\alpha, h)$ and $\phi (g, x)$ belong to $\bigcap_{1\leq p<\infty}L^{p}(Q, \Sigma, {\rm d} \phi_{C})$.  

\goodbreak

\begin{lemma}[Integrals over sharp-time fields \cite{GeJ2}]
\label{LM1}
\quad \begin{itemize}
\item [$ i.)$]  For $h\in L^{1}(\mathbb{R})\cap L^{2}(\mathbb{R})$ and $\alpha \in [0, 2\pi)$ the  limit 
	\begin{equation}
		\label{vr} 
		\lim_{\kappa \to \infty}
		\int_{\mathbb{R}}h(x):\!\phi(\alpha, \delta_{\kappa}(.-x))^{n}\!:_{C_{0}}{\rm d} x  
	\end{equation}
exists in $\bigcap_{1\leq p<\infty}L^{p}(Q, \Sigma, {\rm d}\phi_{C})$. Denote it by 
$\int_{\mathbb{R}} h(x) :\!\phi(\alpha, x)^{n}\!:_{C_{0}}{\rm d} x $. Normal ordering in (\ref{vr})  is with respect to   
the {\em temperature~$\beta^{-1}$ covariance} on $\mathbb{R}$: for $h_{1}, h_{2}\in {\mathscr S}(\mathbb{R})$
	\begin{equation}
		\qquad C_{0}(h_{1}, h_{2}):=  \Bigl( h_{1}, \frac{(1+{\rm e}^{-\beta\epsilon})}{
		2\epsilon(1-{\rm e}^{- \beta \epsilon})}h_{2} \Bigr)_{L^{2}(\mathbb{R}, {\rm d} x)} \, . 
		\label{c0}
	\end{equation}
\item [$ ii.)$]   For $g\in L^{1}(S_{\beta})\cap L^{2}(S_{\beta})$ and $x \in \mathbb{R}$
the limit
	\begin{equation}
		\label{vc} \lim_{k \to \infty}
		\int_{S_{\beta}}g(\alpha):\!\phi(\delta_{k}(.-\alpha), x)^{n}\!:_{C_{\beta}}{\rm d} \alpha   
	\end{equation}
exists in $\bigcap_{1\leq p<\infty}L^{p}(Q, \Sigma, {\rm d}\phi_{C})$. Denote it by 
$\int_{S_{\beta}} g(\alpha) \, {:}\phi(\alpha, x)^{n}{:}_{C_{\beta}}{\rm d} \alpha $. Normal ordering 
in (\ref{vc}) is w.r.t.~the covariance
	\begin{equation}
		\label{cbeta}
		C_{\beta}(g_{1}, g_{2}):=\Bigl(g_{1}, \frac{1}{2\nu}g_{2}\Bigr)_{L^{2}( S_{\beta}, {\rm d} \alpha)}, 
		\qquad g_{1}, g_{2} \in {\mathscr S}(S_{\beta}).
	\end{equation}
\end{itemize}
\end{lemma}

Returning to  the integral in~(\ref{uvren}), we let $f$ be the characteristic function on $S_{\beta} \times [-l,l]$. 
This enables us to rewrite (\ref{uvren}) as $\lim_{k,\kappa \to \infty} F(k,\kappa),$ where
	\[
		F(k,\kappa) 
		=  \sum_{m=0}^{[n/2]} \frac{n! \big(-\frac{1}{2} 
		C(\delta^{(2)}_{k,\kappa},\delta^{(2)}_{k,\kappa}) \big)^m}{m!(n-2m)!} 
		\int_{S_{\beta} \times [-l,l]} \kern -1cm {\rm d} \alpha \,
		{\rm d} x \; \; \phi\big(\delta_k(\cdot - \alpha) \otimes 
		\delta_{\kappa}(\cdot -x) \big)^{n-2m}  ,
	\]
and $\delta^{(2)}_{k,\kappa}(\alpha,x) := \delta_k(\alpha) \otimes \delta_{\kappa}(x).$ 
Interchanging integrals and limits is permitted by the existence of the limits in (\ref{uvren}), 
(\ref{vr}) and (\ref{vc}).  Performing the two limits in different orders results  in	
	\[
		\lim_{k,\kappa \to \infty} F(k,\kappa) = \lim_{\kappa \to \infty} \sum_{m=0}^{[n/2]} 
		\frac{n! \big(-\frac{1}{2} C_{0}(\delta_{\kappa},\delta_{\kappa}) \big)^m}{m!(n-2m)!} 
		\int_{S_{\beta}} \kern -.2cm  {\rm d} \alpha \int_{ [-l,l]} \kern -.5cm {\rm d} x  \; \; 
		\phi\big( \alpha, \delta_{\kappa}(\cdot - x) \big)^{n-2m} 
	\]
and 
	\[
		\lim_{k,\kappa \to \infty} F(k,\kappa) = \lim_{k \to \infty} \sum_{m=0}^{[n/2]} 
		\frac{n! \big(-\frac{1}{2} C_{\beta}(\delta_{k},\delta_{k}) \big)^m}{m!(n-2m)!} 
		\int_{ [-l,l]} \kern -.5cm {\rm d} x
		\int_{S_{\beta}} \kern -.2cm  {\rm d} \alpha \; \; \phi\big(\delta_k(\cdot - \alpha),
		 x \big)^{n-2m} .
	\]
Note that in the latter expression  normal ordering is done w.r.t.~the covariance~$C_{\beta}$, whilst in the former 
normal ordering is done with respect to  the  temperature~$\beta^{-1}$ covariance~$C_0$ on~$\mathbb{R}$. 

Now let $U(\alpha, x)$, with $\alpha \in [0, 2\pi)$ and $x \in {\mathbb R}$, denote the unitary 
operators implementing the rotations and translations on the cylinder in $L^2 ( Q, \Sigma, {\rm d} \mu)$ 
(for further details see next section). It follows that the $L^1$-function (\ref{co}) equals
	\begin{equation}
		{\rm e}^{-\int_{-l}^{l} U(0, x) ( \int_{-\beta/2}^{\beta/2} {:} 
		 {\mathscr P}(\phi(\alpha, 0)) {:}_{C_{\beta}} \; {\rm d} \alpha ) {\rm d} x} 
		= {\rm e}^{-\int_{-\beta/2}^{\beta/2} U(\alpha,0) ( \int_{-l}^{l} 
		{:} {\mathscr P}(\phi(0, x)) {:}_{C_{0}} \; {\rm d} x ) {\rm d} \alpha} .
		\label{ns}
	\end{equation}
A proof of this identity can be found in  \cite[Lemma 5.3]{GeJ2}. The analog of (\ref{ns}) in the 
case $\beta = \infty$ is known as {\it Nelson 
symmetry} (see \emph{e.g.}~\cite{Si1}). Interpreting $x$ in~(\ref{ns}) as the imaginary time one notices that
${\rm d}\mu=\lim_{l\to \infty}{\rm d} \mu_l$ is the Euclidean measure of the vacuum 
${\mathscr P}(\varphi)_2$ model on the circle. This argument can be
made rigorous (see \cite[Theorem 7.2]{GeJ2}, \cite{H-K}) by exploiting various properties of a time 
dependent heat equation (see 
\cite[Appendix~A]{GeJ2}).

\goodbreak
\begin{theorem}\label{3.2} 
Consider sharp-time fields as defined in (\ref{e1.4}), and integrals over normal ordered 
products as defined in (\ref{vr}) and (\ref{vc}). 
\begin{itemize}
\item[$i.)$]
{\rm  (Thermodynamic limit of Euclidean measures).}
For $f \in C^\infty_{0\mathbb{R}}(S_{\beta}\times\mathbb{R})$
\begin{equation}
\label{limit3}
E_{\mathscr P} (f) = \lim_{l\to +\infty} \frac{1}{Z_l} \int_{Q}
{\rm e}^{ i \phi(f)}
\, {\rm e}^{-\int_{-l}^{l} U (0, x) \left( \int_{-\beta/2}^{\beta/2} {:} P(\phi(\alpha, 0)) {:}_{C_{\beta}} \; 
{\rm d} \alpha \right) {\rm d} x} {\rm d} \phi_{C}  \; . 
\end{equation}
\item [$ii.)$]{\rm (Nelson symmetry).}
For $f \in C^\infty_{0\mathbb{R}}(S_{\beta}\times\mathbb{R})$
\begin{equation}
\label{limit4}
E_{\mathscr P} (f) = \lim_{l\to +\infty} \frac{1}{Z_l} \int_{Q}
{\rm e}^{ i \phi(f)} \, {\rm e}^{-\int_{-\beta/2}^{\beta/2}  U(\alpha, 0) 
\left( \int_{-l}^{l} {:}  P(\phi(0, x)) {:}_{C_{0}} \; {\rm d} x \right) {\rm d} \alpha} {\rm d} \phi_{C} \; .
\end{equation}
\end{itemize}
The map $f\mapsto E_{\mathscr P}(f)$ is continuous in some Schwartz semi-norm and thus
extends to ${\mathscr S}(S_{\beta}\times \mathbb{R})$ \cite[Theorem 7.2 $ii.)$]{GeJ2}. It satisfies 
the conditions of Minlos' theorem and thus defines a probability measure ${\rm d} \mu$. 
\end{theorem}

\paragraph{Remark}
This result solves the {\em infrared problem} for the thermal field theory under consideration. As mentioned 
before, we could have used cluster expansions  to resolve this problem. However,
Nelson symmetry will play a key role in the sequel, enabling us to transfer results between the two models
it connects.

\smallskip
Before we continue, we recall two results, which refer to the $L^p$-spaces for the interacting
measure~${\rm d} \mu$: 

\goodbreak
\begin{lemma}{\rm \cite[Propositions 7.3 and 7.5]{GeJ2}}
\label{3.3}
\begin{itemize}
\item[$ i.)$] {\rm (Sharp-time fields are in $L^{p}(Q, \Sigma, {\rm d} \mu)$)}.
Let $h\in {\mathscr S}_{\mathbb{R}}(\mathbb{R})$ and $\alpha \in S_{\beta}$. Then the sequence
$\phi \bigl(\delta_{k}(.-\alpha)\otimes h \bigr)$ is Cauchy 
in~$\bigcap_{1\leq p<\infty}L^{p}(Q, \Sigma, {\rm d} \mu)$ and hence
	\[
		\phi(\alpha, h):= \lim_{k\to \infty}\phi \bigl(\delta_{k}(\, .\, -\alpha)\otimes h \bigr)\in
		\bigcap_{1\leq p<\infty}L^{p}(Q, \Sigma, {\rm d} \mu).
	\]
Moreover, the map
	\[
		\begin{array}{rcc}
		S_{\beta} & \to &  \bigcap_{1\leq p<\infty}L^{p}(Q, \Sigma, {\rm d} \mu)\\
		\alpha & \mapsto & \phi(\alpha, h)
		\end{array}
	\]
is continuous for $h\in {\mathscr S}_{\mathbb{R}}(\mathbb{R})$ fixed.
\item[$ii.)$] {\rm (Convergence of sharp-time Schwinger functions, Part I)}.
	\label{T3}
	Let $h_{i}\in C^{\infty}_{0\: \mathbb{R}}(\mathbb{R})$ and $\alpha_{i}\in S_{\beta}$, $1\leq i\leq n$.
	Then
		\[
			\lim_{l\to \infty}\int_{Q} \Bigl( \prod_{j=1}^{n}{\rm e}^{i \phi(\alpha_{j},
			h_{j})} \Bigr) {\rm d}\mu_{l}= \int_{Q} \Bigl( \prod_{j=1}^{n}{\rm e}^{i \phi(\alpha_{j},
			h_{j})} \Bigr) {\rm d}\mu.
		\]
\end{itemize}
\end{lemma}

In Section \ref{sec:3.2} we will show that products of Euclidean sharp-time fields are as well elements of 
$\bigcap_{1\leq p<\infty}L^{p}(Q, \Sigma, {\rm d} \mu)$. This will allow us to extend results 
of Fr\"ohlich \cite{Fr1}, Fr\"ohlich and Birke \cite{BF}, and Klein and Landau \cite{KL1}
concerning the reconstruction of thermal Green's functions.

\subsection{The Osterwalder-Schrader Reconstruction}
\label{sectosrec}

The cylinder $S_\beta \times \mathbb{R}$ is invariant under  rotations and translations
	\[ 
		\mathfrak{t}_{({\alpha'},x')} \colon (\alpha,x) \mapsto
		(\alpha+{\alpha'},x+x') ,  \qquad {\alpha'} \in [0,2\pi), \quad x' \in \mathbb{R}, 
	\]
as well as the reflections $\mathfrak{r}\colon(\alpha,x) \mapsto (-\alpha,x)$ and $\mathfrak{r}'\colon (\alpha,x) \mapsto (\alpha,-x)$. 
The pull-backs 
	\[
		(\mathfrak{t}^{({\alpha'},x')}_* f )(\alpha, x)
		:= f\left( \mathfrak{t}_{ ({\alpha'},x')}^{-1}  (\alpha,x) \right)= f (\alpha- {\alpha'}, x-x') 
	\]
acting on the testfunctions $f \in {\mathscr S} (S_{\beta} \times \mathbb{R})$, induce  actions on the tempered distributions  $q \in Q$:
	\[
		(t_{({\alpha'},x')}q)(f) := \langle q, \mathfrak{t}^{(-{\alpha'},-x')}_* f \rangle, \quad (rq)(f) 
		:= \langle q, \mathfrak{r}_* f \rangle, \quad \textrm{and}
		\quad (r'q)(f) := \langle q, \mathfrak{r}'_* f \rangle. \nonumber
	\]
Lifting these maps to measurable functions of distribution one finds that
	\begin{itemize}
		\item [$i.)$] the map $U(\alpha,x)F(q):=F(t_{(\alpha,x)}^{-1}q)$,  $q \in Q $,  
			defines a two-parameter group of measure-preserving $*$-automorphisms of $L^{\infty}(Q,\Sigma, {\rm d} \mu)$, strongly
			continuous in measure, and strongly continuous two-parameter groups of isometries of~$L^{p}(Q,\Sigma, {\rm d} \mu)$ 
			for $1 \le p < \infty$; 
		\item [$ii.)$] the maps $ RF(q):=F(rq)$ and $R'F(q):= F(r'q)$ extend to two  measure preserving $*$-auto\-morphisms 
			of~$L^{\infty}(Q,\Sigma, {\rm d} \mu)$ and to isometries of~$L^{p}(Q,\Sigma, {\rm d} \mu)$ for $1 \le p < \infty$. 
	\end{itemize}
Since ${\rm d} \mu$ is translation and rotation invariant, $U(\gamma,y)$ is unitary on the Hilbert space $L^{2}(Q,\Sigma, {\rm d} \mu)$ 
for $\gamma \in [0,\beta)$ and $y \in \mathbb{R}$.

\bigskip

\noindent
{\em Notation}. For $0 \le \gamma \leq \beta$ (resp.~$0 \le y \le \infty$)  we denote by $\Sigma_{[0, \gamma]}$  (resp.~$\Sigma^{[0, y]}$) 
the sub $\sigma$-algebra of the Borel $\sigma$-algebra $\Sigma$ generated by the functions ${\rm e}^{i \phi(f)}$ with 
$f \in {\mathscr S}_{\mathbb{R}} (S_\beta \times \mathbb{R})$ and ${\rm supp} \;  f\subset [0, \gamma ] \, \times\mathbb{R}$ 
(resp.~${\rm supp} \;  f\subset S_\beta \times[0, y]$). 

\bigskip

Next define two scalar products:
	\[
		\forall F,G \in L^{2}(Q, \Sigma_{[0, \beta/2]}, {\rm d} \mu):  \quad  (F, G):=\int_{Q}R(\overline{F})G \; {\rm d}\mu , 
	\]
and
	\[
		\forall F,G \in L^{2}(Q, \Sigma^{[0, \infty)}, {\rm d}\mu):  \quad  (F, G)':=
		\int_{Q}R'(\overline{F})G \; {\rm d}\mu  . 
	\]
The measure ${\rm d}\mu$ is {\em Osterwalder-Schrader positive} with respect to~{\em both} reflections~$R$ and~$R'$:
	\[
		\forall F \in  L^{2}(Q, \Sigma_{[0, \beta/2]}, {\rm d} \mu):\,(F,F) \geq 0 
	\]
and
	\[ 
		\forall G \in L^{2}(Q, \Sigma^{[0, \infty)}, {\rm d}\mu):\,(G,G)'\geq 0.
	\]
Let ${\cal N}\subset L^{2}(Q, \Sigma_{[0, \beta/2]}, {\rm d} \mu) $ be the kernel of the positive quadratic form $(\, .\, ,\, .\, )$
and  ${\cal N}' \subset L^{2}(Q, \Sigma^{[0, \infty)}, {\rm d} \mu) $  the kernel of the positive quadratic form~$(\, .\, , \, .\, )'$.  Set
	\[
		{\mathcal H}_\beta:= \overline{ L^2 (Q, \Sigma_{ [0, \beta/2] }, {\rm d} \mu)  / {\cal N}  } 
		\quad
		\hbox{and}
		\quad
		{\mathcal H}_C:= \overline{ L^2(Q, \Sigma^{[0, \infty)}, {\rm d} \mu)  / {\cal N}'  }.
	\]
The completions of the pre-Hilbert spaces are taken w.r.t.~the norms $(\, .\, , \, . \,)^{\frac{1}{2}}$ 
and~${(\, .\, , \, .\, )'}^{\frac{1}{2}} $, respectively. The canonical projection from $L^{2}(Q, \Sigma_{[0, \beta/2]}, {\rm d} \mu) $ 
to~${\mathcal H}_\beta$ and from $L^{2}(Q, \Sigma^{[0, \infty)}, {\rm d} \mu) $ to~${\mathcal H}_C$
are denoted by ${\cal V}$ and ${\cal V}'$, respectively. The distinguished vectors
	\[
		\Omega_\beta:= {\cal V}(1), \qquad \Omega_C:= {\cal V}'(1), 
	\]
arise as the image of $1$, the constant function equal to $1$ on $Q$. 

\goodbreak
The abelian algebra 
\begin{itemize}
	\item [$i.)$] $L^{\infty}(Q, \Sigma_{\{0\}}, {\rm d} \mu )$  preserves $L^2 (Q, \Sigma_{ [0, \beta/2] }, {\rm d} \mu)$ and ${\cal N}$.
			Thus a representation~$\pi_\beta$ of $L^{\infty}(Q, \Sigma_{\{0\}}, {\rm d} \mu )$ on the Hilbert spaces ${\mathcal H}_\beta$ 
			is given by 
			\[
				\pi_\beta(A){\cal V}(F):= {\cal V}(AF), \quad F\in L^2 (Q, \Sigma_{ [0, \beta/2] }, {\rm d} \mu),
				\quad A \in L^{\infty}(Q, \Sigma_{\{0\}}, {\rm d} \mu );
			\]
	\item [$ii.)$] $L^{\infty}(Q, \Sigma^{\{0\}}, {\rm d} \mu )$ preserves $L^2 (Q, \Sigma^{ [0, \infty) }, {\rm d} \mu)$ and ${\cal N}'$. 
			Thus one obtains a representation~$\pi_C$ of $L^{\infty}(Q, \Sigma^{\{0\}}, {\rm d} \mu )$ on  ${\mathcal H}_C$, specified by 
			\[
				\pi_C(B){\cal V}'(G):= {\cal V}'(BG), \quad G\in L^2(Q, \Sigma^{[0, \infty)}, {\rm d} \mu),
				\quad B \in L^{\infty}(Q, \Sigma^{\{0\}}, {\rm d} \mu ). 
			\]
\end{itemize}
The corresponding von Neumann algebras can be interpreted as the algebras generated by bounded functions of the thermal time-zero 
fields on the real line and the vacuum time-zero fields on the circle, respectively.  

\bigskip
The reconstruction of the dynamics requires a more pronounced distinction of the two cases under consideration, which in the 
thermal case relies on a remarkable result on local symmetric semi-groups by Fr\"ohlich \cite{Fr2} and, independently, Klein and 
Landau \cite{KL2}:

\begin{itemize}
	\item[$ i.)$] The semigroup $\{ U(\alpha,0)\}_{\alpha>0}$ does \emph{not} preserve
		$L^2 (Q, \Sigma_{ [0, \beta/2] }, {\rm d} \mu)$. But setting, for $0\leq \gamma \leq \beta/2$, 
			\[ 
				{\mathcal D}_{\gamma}:={\cal V}{\cal M}_{\gamma}, \qquad \hbox{with} \quad
				{\cal M}_{\gamma}: =L^{2}(Q, \Sigma_{[0, \beta/2 -\gamma]}, {\rm d} \mu) , 
			\]
		one can define, for $0\leq \alpha \leq \gamma$, a linear operator $P(\alpha) \colon {\cal D}_{\gamma}\to {\mathcal H}_\beta$ 
		with domain~${\mathcal D}_{\gamma}$ by setting
			\[
				P(\alpha){\cal V}\psi:= {\cal V}U(\alpha, 0)\psi, \qquad  \psi\in {\cal M}_{\gamma}.
			\]
		The triple  $(P(\alpha), {\mathcal D}_{\alpha}, \beta/2)$ forms a {\em local symmetric semigroup} (see \cite{Fr2}\cite{KL2}):

\smallskip 
\begin{itemize}
	\item[$a.)$] for each $ \alpha$, $0 \le \alpha \le \beta/2$, ${\cal D}_{\alpha}$ is a linear subset of ${\mathcal H}_\beta$ 
		such that ${\cal D}_{\alpha} \supset {\cal D}_{\gamma}$ if $0 \le \alpha \le \gamma \le \beta /2$, and 
			\[ 
				{\mathcal D} := \bigcup_{0 < \alpha \le \beta/2} {\cal D}_{\alpha}
			\] 
		is dense in ${\mathcal H}_\beta$;
	\item[$b.)$] for each $ \alpha$, $0 \le \alpha \le \beta/2$, $P(\alpha)$ is a linear operator on  ${\mathcal H}_\beta$ 
		with domain~${\cal D}_{\alpha}$;
	\item[$c.)$] $P(0) = 1$, $P(\alpha) {\mathcal D}_\gamma \subset {\mathcal D}_{\gamma - \alpha}$ for $0 \le \alpha \le \gamma \le \beta/2$, and 
			\[ 
				P(\alpha) P(\gamma)= P(\alpha+\gamma)
			\]
		on ${\mathcal D}_{\alpha + \gamma} $ for $\alpha, \gamma, \alpha + \gamma \in [0, \beta / 2]$;
	\item[$d.)$] $P(\alpha)$ is symmetric, \emph{i.e.}, 
			\[ 
				(\Psi, P(\alpha) \Psi' ) =  (P(\alpha) \Psi' , \Psi ), \qquad  0 \le \alpha \le \beta / 2, 
			\]
		for all $\Psi, \Psi' \in {\mathcal D}_\alpha$ and $0 \le \alpha \le \beta / 2$;
	\item[$e.)$] $P(\alpha)$ is weakly continuous, \emph{i.e.}, if $ \Psi \in {\mathcal D}_\gamma$,  $0 \le \gamma \le \beta/2$, then 
			\[ 
				\alpha \mapsto   (\Psi, P(\alpha) \Psi ) 
			\]
		is a continuous function of $\alpha$ for $0 \le \alpha \le \gamma$.  
\end{itemize}

By the results cited \cite{Fr2}\cite{KL2}  there exists a selfadjoint operator~$L$ 
on~${\mathcal H}_\beta$ such that for $0\le \alpha \le \gamma  $ 
	\[
		{\cal V}(U(\alpha,0)F)= {\rm e}^{-\alpha L }{\cal V}(F), \qquad F\in L^{2}(Q, \Sigma_{[0, \beta/2-\gamma]},
		{\rm d}\mu).
	\]
The selfadjoint operator $L$  is said to be {\em associated }to the local symmetric semigroup 
$(P(\alpha), {\mathcal D}_{\alpha}, \beta/2)$.
Since $1 \in {\cal M}_{\gamma}$ and 
$L^{\infty}(Q, \Sigma_{\{0\}}, {\rm d} \mu ){\cal M}_{\gamma} \subset {\cal M}_{\gamma}$  
 for all $0 \le \gamma \le \beta/2$, it follows that 
${\rm e}^{i\phi_\beta(h)} \Omega_\beta \in {\mathcal D}( {\rm e}^{- \frac{\beta}{2} L})$, 
where ${\rm e}^{i\phi_\beta(h)} \doteq \pi_\beta ({\rm e}^{i\phi(0,h)}) $ with 
$h \in C_{0\mathbb{R}}^\infty  (\mathbb{R})$. 

\begin{lemma} 
\label{Lm3}
$ {\mathcal D}_\gamma$ is dense in ${\mathcal H}_\beta$ for  $0< \gamma < \beta/2$. 
\end{lemma}

\begin{proof} Assume that 
	\begin{equation}
		\label{f23}
		( \Psi , \Phi) = 0 \qquad \forall \Phi \in {\mathcal D}_\gamma . 
	\end{equation}
Now consider, for $h_1, h_2 \in C_{0\mathbb{R}}^\infty  (\mathbb{R})$ fixed,  the analytic function
	\begin{equation}
		\label{f24} z \mapsto 
		( \Psi , {\rm e}^{i \phi_\beta(h_1)} 
		{\rm e}^{-zL} {\rm e}^{i\phi_\beta(h_2)} \Omega_\beta) , \qquad \{ z \in \mathbb{C} \mid 0 < \Re z  < \beta/2\} . 	
		\end{equation} 
Clearly, ${\rm e}^{i \phi_\beta(h_1)} {\rm e}^{-\Re z L} {\rm e}^{i\phi_\beta(h_2)} \Omega_\beta \in {\mathcal D}_\gamma$ 
for $0 < \Re z  < \gamma$ and consequently, because of (\ref{f23}),
the analytic function (\ref{f24}) vanishes on an open line segment in the interior of its domain, and is therefore identically zero. 
It follows that 
	\begin{equation}
		\label{f25}
		( \Psi , {\rm e}^{i \phi_\beta(h_1) } {\rm e}^{-\frac{\beta}{2} L} {\rm e}^{i \phi_\beta(h_2) } \Omega_\beta) 
		= 0 \qquad \forall h_1, h_2 \in C_{0\mathbb{R}}^\infty  (\mathbb{R})  . 
	\end{equation}
The set $\{ {\rm e}^{i \phi_\beta(h_1)} {\rm e}^{-\frac{\beta}{2}L} {\rm e}^{i\phi_\beta(h_2) }\Omega_\beta \mid h_1, h_2 
\in C_{0\mathbb{R}}^\infty  (\mathbb{R}) \}$ is dense in ${\mathcal H}_\beta$ \cite[Theorem 11.2]{KL1}, and 
therefore (\ref{f25}) implies $\Psi =0$. In other words, ${\mathcal D}_\gamma$ is dense in~${\mathcal H}_\beta$.  
% \qed
\end{proof}

\item[$ii.)$]  The semi-group $U(0, x)$, $x \geq 0$,  preserves the half-space $L^2 (Q, \Sigma^{ [0, \infty) }, {\rm d} \mu)$ as 
$U(0, x)$ maps $L^{2}(Q, \Sigma_{[0, \infty)}, {\rm d}\mu)$ into itself. Following \cite{K} one can therefore define a self-adjoint 
positive operator $H_C$ on ${\mathcal H}_C$ such that for $G\in L^{2}(Q, \Sigma_{[0, \infty)}, {\rm d}\mu)$ 
	\begin{equation}
		\label{hc}
		{\cal V}'(U(0,x)G)= {\rm e}^{-x H_C }{\cal V}'(G), \qquad x >0 .
	\end{equation}
The operators ${\rm e}^{-x H_C }$, $x >0$, form a strongly continuous semigroup of contractions
on ${\mathcal H}_C$. 

\end{itemize}
 
\bigskip

The next step in  the reconstruction program is to define non-abelian von Neumann algebras 
${\mathcal R}_\beta \subset {\mathscr B}({\mathcal H}_\beta)$ and ${\mathcal R}_C \subset {\mathscr B}({\mathcal H}_C)$, 
generated by the operators 
	\[ 
		\tau_{t,0} (\pi_\beta(A)):= {\rm e}^{i tL}\pi_\beta(A){\rm e}^{-i tL}, 
		\quad t\in \mathbb{R}, \quad A\in L^{\infty}(Q, \Sigma_{\{0\}}, {\rm d} \mu) , 
	\]
and
	\[ 
		\tau'_{0,\sigma} (\pi_C(A)):= {\rm e}^{i \sigma H_C}\pi_C(A){\rm e}^{-i \sigma H_C}, 
		\quad \sigma\in \mathbb{R}, \quad A\in L^{\infty}(Q, \Sigma^{\{0\}}, {\rm d} \mu) , 
	\]
respectively. Clearly $\tau_{t,0}$ and  $\tau'_{0,\sigma}$ extend to  *-automorphisms of~${\mathcal R}_\beta$ 
and ${\mathcal R}_C$, respectively. 

The algebra ${\mathcal R}_\beta \subset {\mathscr B}({\mathcal H}_\beta)$ has a cyclic and separating vector, 
namely~$\Omega_\beta$. The time-translation invariant state $\omega_\beta$ (a normalised positive linear functional) 
on~${\mathcal R}_\beta$ defined by 
	\[ 
		\omega_\beta (a) := (\Omega_\beta , a \, \Omega_\beta ), \quad a \in {\mathcal R}_\beta , 
	\]
is invariant under the spatial translations induced by $ \mathfrak{t}_{(0,y)} $, $ y \in \mathbb{R}$. Furthermore, it 
satisfies the {\em KMS condition} \cite{KL1}: the functions 
	\[
		F_{h_1, \ldots, h_n}  (t_1-t_2, \ldots, t_{n-1}-t_{n}):= \bigl(\Omega_\beta \, , \tau_{t_1} ( {\rm e}^{i \phi_\beta(h_1)}) \ldots 
		\tau_{t_{n}} ({\rm e}^{i \phi_\beta(h_n)} ) \Omega_\beta \bigr) 
	\] 
extend to analytic functions in the domain
	\[
		\textstyle { \left\{ (z_1, \ldots ,  z_{n-1}) \in \mathbb{C}^{n-1} \mid \Im z_k < 0, \; - \beta < \sum_{k=1}^{n-1} \Im z_k  \right\} }
	\]
and satisfy the KMS boundary condition: for each $1 \le k <n$ 
	\begin{align}
		\label{KMSWEYL}
		&F_{h_1, \ldots, h_n}  ( s_1, \ldots, s_{k-2} , 
		s_{k-1} - i \beta , s_{k}, \ldots, s_{n-1} ) 
		\nonumber
		\\ 
		& \qquad  = F_{ h_k, \ldots, h_n, h_1, \ldots, 
		h_{k-1}}  (s_{k} , \ldots, s_{n-1}   , s_n  , s_1 , \ldots,  s_{k-2} )  
	\end{align}
with  $s_n = t_n - t_1$  and $s_k = t_k - t_{k+1}$, $k= 1, \ldots, n-1$,  
and $h_1, \ldots , h_n \in C_{0 \mathbb{R}}^\infty$. 
  
The algebra ${\mathcal R}_C \subset {\mathscr B}({\mathcal H}_C)$ has a cyclic vector, namely $\Omega_C$. 
The state~$\omega_C$ on~${\mathcal R}_C$, 
	\[ 
		\omega_C (a) := (\Omega_C , a \, \Omega_C ), \quad a \in {\mathcal R}_C , 
	\]
is invariant under the rotations induced by $ \mathfrak{t}_{(\gamma,0)}$, $ \gamma \in [0,2\pi)$, and satisfies the {\em spectrum condition} 
(see Theorem \ref{HO} below), which characterises {\em vacuum states}. Since $\omega_C$ is the unique vacuum state (see below), 
the commutant~${\mathcal R}_C'$ of ${\mathcal R}_C$ equals $\mathbb{C} \cdot 1$ and therefore ${\mathcal R}_C = {\mathscr B}({\mathcal H}_C)$. 

\subsection{The Wightman functions on the Einstein universe}
\label{SSec2.4}

The Hilbert space ${\mathcal H}_C$ reconstructed in the previous section is unitarily equivalent  to the Fock 
space $\Gamma \bigl(H^{-\frac{1}{2}}(S_{\beta}) \bigr)$ over the Sobolev space~$H^{-\frac{1}{2}}(S_{\beta})$ 
of order $-\frac{1}{2}$ on~$S_{\beta}$, equipped with the norm
	\[
		\|g\|^{2}
		= \left(g, (2\nu)^{-1}g \right)_{L^{2}(S_{\beta}, {\rm d} \alpha)} , \qquad 
		\nu= \left(D_{\alpha}^{2}+m^{2}\right)^{\frac{1}{2}} .
	\]
To ease the notation we simply identify corresponding operators and vectors.
For $g \in {\mathscr S}_{\mathbb{R}}(S_{\beta})$ the Segal field operator on ${\mathcal H}_C$, given by 
	\[ 
		\phi_C  (g) := - i \frac{{\rm d}}{ {\rm d} \lambda} {\cal V}' ({\rm e}^{i \phi(0,\lambda g)} )\Bigl|_{ \lambda=0} \;  , 
	\]
is thereby identified with the Fock space field operator 
	\begin{equation} 
		\label{creat-ann}
		\phi_C  (g)  = {1 \over \sqrt{2}}\bigl ( { a}^*(\nu^{-1/2}{g}) + { a}(\nu^{-1/2}{g}) \bigr)^- \; , 
	\end{equation}
built up from bosonic creation and annihilation operators ${ a}^*({f}) $ and ${ a}({f})$ (see, \emph{e.g.},
\cite{RS}). 
Note that the map $f \mapsto a^*(f)$ is linear, while the map $f \mapsto a(f)$ is anti-linear.

The (angular) momentum operator $P_{C}:={\rm d}\Gamma(D_{\alpha})$ on the circle $S_\beta$ 
has discrete spectrum. Define 
	\[ 
		V:= \int_{S_{\beta}} {:} {\mathscr P}(\phi_C (\alpha)){:}_{C_{\beta}}{\rm d} \alpha \, .
	\] 
The operator sum
	\[ 
		{\rm d}\Gamma(\nu) + V - E_C
	\]
is essentially selfadjoint on its natural domain ${\mathcal D}({\rm d}\Gamma(\nu)) \cap {\mathcal D}(V)$ and 
bounded from below. Its closure equals the Hamiltonian $H_{C}$ of the ${\mathscr P}(\phi_C)_{2}$ model on 
the circle $S_{\beta}$, which has been (re-)constructed in the previous section (see (\ref{hc})). The additive 
constant $E_{C}$ is chosen such that zero is the lowest eigenvalue, \emph{i.e.}, ${\rm inf\ }{\rm Spec} \, (H_{C})=0$. This 
eigenvalue is non-degenerated\footnote{Glimm and Jaffe have shown 
in \cite{Glimm-Jaffe-Uniqueness} 
that the Hamiltonian $H$ with a spatial cutoff, rather than on a spatial circle, \emph{i.e.}, with periodic boundary 
conditions, satisfies the properties stated in this paragraph.  Similar arguments apply to $H_C$, see the proof 
of Proposition 5.4 in \cite{GeJ2}.}, and the corresponding eigenvector  $\Omega_{C}$ can be 
chosen such that 
$ (\Omega_{C}, \Omega^\circ) >0$. Here $\Omega^\circ$ denotes the Fock vacuum vector in 
$\Gamma \bigl(H^{-\frac{1}{2}}(S_{\beta}) \bigr)$.

Moreover, the Glimm-Jaffe $\phi$-bounds (see  \emph{e.g.}~\cite{Fr5}\cite{GJ-IV}\cite{GRS}, the  exact variant we 
use can be found in \cite[Proposition 5.4]{GeJ2}) hold: for $c > \kern -.15cm > 1$ and some 
${\tt C} \in \mathbb{R}^+$, 
	\begin{equation}
		\label{Fockcontinuity1}
		\pm \phi_C (g) \le  {\tt C} \, \| g \|_{H^{-\frac{1}{2}} (S_{\beta})} (H_C + c)^{1/2}  
		\qquad \forall g \in H^{-\frac{1}{2}}(S_{\beta})\; , 
	\end{equation}
and
	\begin{equation}
		\label{Fockcontinuity2}
		\pm \phi_C (g) \le  {\tt C} \, \| g \|_{H^{-1} (S_{\beta})} (H_C + c)   
		\qquad \forall g \in H^{-1}(S_{\beta})\; . 
	\end{equation}

\bigskip
The following remarkable result is due to Heifets \& Osipov \cite{HO}; see also 
\cite{JJM}.

\begin{theorem}[Spectrum Condition \cite{HO}]
\label{HO}
The joint spectrum of $P_C$ and $H_C$ is purely discrete and contained 
in the forward light cone $\widetilde{V}^+  := \{ (p, E) \mid | p | < E  \} $.
\end{theorem}

\noindent
The unitary operators $U_C (\alpha, \sigma) \in {\mathscr B}({\mathcal H}_C)$ given by 
	\begin{equation}
		\label{uc}
			U_C (\alpha, \sigma):= {\rm e}^{i (\sigma H_C - \alpha P_C) },  \qquad \alpha \in [ 0, 2 \pi ), \; \sigma \in \mathbb{R}, 
	\end{equation}
implement the two parameter group of automorphisms $\tau'_{\alpha,\sigma}$ of ${\mathcal R}_C$ on the Hilbert space ${\mathcal H}_C$. 
Let~$g_i \in {\mathscr S}(S_{\beta})$ and set
	\[ 
		\phi_{C}(g_i, \sigma_i):= {\rm e}^{i \sigma_i H_C}\phi_{C}(g_i) {\rm e}^{- i \sigma_i H_C} , \quad i = 1, \ldots, n. 
	\]
By Stone's theorem, the map $\sigma \mapsto U_C (0, \sigma)$ is strongly continuous. 
Together with the bound~(\ref{Fockcontinuity1}) this implies that  
	\[
		{\mathscr W}_{C}^{(n)}( g_1, \sigma_1,  \ldots, g_n, \sigma_{n}) 
		:= \bigl(\Omega_{C}, \phi_{C}(g_1, \sigma_1)\cdots \phi_{C}(g_n, \sigma_n) \Omega_{C}\bigr) 
	\]
exists and is a separately continuous multi-linear functional of the arguments 
$(g_i, \sigma_i)$, $i = 1, \ldots, n$,  as they vary over ${\mathscr S} (S_\beta) \times \mathbb{R}$. It follows from the
nuclear theorem \cite{SW} that this functional can be uniquely represented as a tempered distribution of the $n$ vectors 
$(\alpha_i, \sigma_i) \in S_\beta \times \mathbb{R}$. Denote the corresponding distribution  by 
	\begin{equation}
		\label{wc}
		{\mathscr W}_{C}^{(n)}( \alpha_1, \sigma_1,  \ldots, \alpha_n, \sigma_{n})
		\equiv \bigl(\Omega_{C}, \phi_{C}(\alpha_1, \sigma_1)\cdots \phi_{C}(\alpha_n, \sigma_n) \Omega_{C}\bigr). 
	\end{equation}
Translation invariance implies that ${\mathscr W}_{C}^{(n)}$ depends only on the relative coordinates 
	\[ 
		\xi_i = (\alpha_{i} - \alpha_{i+1} , \sigma_{i} - \sigma_{i+1}) , \quad i= 1, \ldots, n-1 , 
	\]
or more precisely, that there exists a tempered distribution ${\mathfrak W}_{C}^{(n-1)}$ such that 
	\begin{equation}
		\label{37}
		{\mathfrak W}_{C}^{(n-1)}( \xi_1, \xi_2, \ldots, \xi_{n-1})= {\mathscr W}_{C}^{(n)}( \alpha_1, \sigma_1,  \alpha_2, \sigma_2,\ldots, \alpha_n, \sigma_{n}). 
	\end{equation}
We interpret ${{\mathfrak W}}_{C}^{(n-1)}$ as a periodic generalised function, and so its continuous Fourier transform is 
a tempered distribution, which can be identified with its discrete Fourier transform.

\begin{lemma} 
\label{fourierwforward}
Let $\widetilde {{\mathfrak W}}_{C}^{(n-1)}$ denote the Fourier transform of~$ {{\mathfrak W}}_{C}^{(n-1)}$. Then 
the distributional support of $\widetilde {{\mathfrak W}}_{C}^{(n-1)}$ is contained in the joint spectrum of $P_C$ and $H_C$.
\end{lemma}

\begin{proof} The Fourier transform of~$ {{\mathfrak W}}_{C}^{(n-1)}$ is
	\begin{align*}
			& \widetilde {{\mathfrak W}}_{C}^{(n-1)}  \big( (p_1, E_1), (p_2, E_2) , \ldots, (p_{n-1}, E_{n-1}) \big) =
			\\ 
			& \qquad =  (2 \pi \beta)^{-(n-1)} 
			\int  {\rm d} \xi_1 \cdots {\rm d} \xi_{n-1} \;  {\rm e}^{i \sum_{j=1}^{n-1} (p_j, E_j) \cdot \xi_j} \; 
						{{\mathfrak W}}_{C}^{(n-1)} (\xi_1 ,  \ldots , \xi_{n-1}) \; , 
	\end{align*}
where
	\[ 
		\begin{array}{rl}
			& 	{{\mathfrak W}}_{C}^{(n-1)} ( \alpha_1-\alpha_2, \sigma_1 - \sigma_2,  \ldots, \alpha_{n-1}-\alpha_n, \sigma_{n-1} - \sigma_{n}) =
			\\ [2mm]
%& \qquad = \bigl(\Omega_{C}, \phi_{C}(\alpha_1, \sigma_1)\cdots \phi_{C}(\alpha_n, \sigma_n) \Omega_{C}\bigr). 
%\\ [2mm]
		& \qquad = \bigl(\Omega_{C}, \phi_{C}(\alpha_1) {\rm e}^{  i (\sigma_2- \sigma_1) H_C} 
				\cdots \phi_{C}(\alpha_{n-1}) {\rm e}^{i (\sigma_n-\sigma_{n-1}) H_C}\phi_{C}(\alpha_n) \Omega_{C}\bigr) \; . 
		\end{array}
	\]
Next insert, as suggested in \cite{SW},  a basis of common eigenfunctions $\Psi_{\epsilon, k}$ of the operators $P_C, H_C$: 
for all $\Phi \in {\mathcal H}_C$ the unitary operators $U_C (\alpha, s)$ defined in (\ref{uc}) can be expressed as
	\[ 
		U_C (\alpha, \sigma) \Phi = \sum_{(k,\epsilon) \in {\rm Sp} (P_C, H_C)}
 		{\rm e}^{i ( \sigma \epsilon - \alpha k ) } \; ( \Psi_{k, \epsilon} , \Phi ) \, \Psi_{k, \epsilon} \, .  
	\]
Now consider, for $\Phi, \Phi' \in {\mathcal H}_C$ fixed, the map
	\begin{align*} 
	( E, p ) \mapsto & \int_{S_\beta} 
		{\rm d} \alpha \int_{\mathbb{R}} {\rm d} \sigma \; {\rm e}^{ -i ( E \sigma - p  \alpha)  } (\Phi', U(\alpha, \sigma) \Phi)   
		\\
		&  \quad = \sum_{(k, \epsilon) \in {\rm Sp} (P_C, H_C)}  
		\int_{\mathbb{R}} {\rm d} \sigma \; {\rm e}^{ -i ( E - \epsilon ) \sigma  } 
		\times   
		\\ 
		& \qquad \qquad \qquad \qquad \times \int_{S_\beta} 
		 {\rm d} \alpha \;  {\rm e}^{i (p-k)  \alpha  }   
		(\Psi_{k, \epsilon} , \Phi)  (\Phi', \Psi_{k, \epsilon} )
		\\ 
		& \quad = \sum_{(k, \epsilon) \in {\rm Sp} (P_C, H_C)} 
		2 \pi \delta( E - \epsilon)
		\delta_{k,p} (\Psi_{k, \epsilon} , \Phi)  (\Phi', \Psi_{k, \epsilon} ) .
	\end{align*}
The sum on the r.h.s.~vanishes, if $( p, E) \notin  {\rm Sp} (P_C, H_C)$. 
This implies that the distributional support of $\widetilde {{\mathfrak W}}_{C}^{(n-1)}$ is contained in the 
joint spectrum of $P_C$ and~$H_C$. 
% \qed
\end{proof}

\begin{theorem} \label{5.1}
For each $n \ge 1$, $\widetilde {{\mathfrak W}}_{C}^{(n-1)} $ has support in $(
\widetilde{V}^+ )^{n-1}$ and 
${{\mathfrak W}}_{C}^{(n-1)} $ is the boundary value of a polynomially bounded
function ${\mathcal W}_+^{(n-1)}$ analytic in the forward tube $(S_{\beta}\times\mathbb{R}-iV^{+})^{n-1}$, where 
$V^+ := \{ (t,x) \in \mathbb{R}^2 \mid |x| < t \}$.
\end{theorem}

\begin{proof}
The support property of $\widetilde {{\mathfrak W}}_{C}^{(n-1)}$ was established in
Lemma \ref{fourierwforward}. By the Bros-Epstein-Glaser Lemma \cite[Theorem IX.15]{RS}  
there exists a polynomial $P$ and a polynomially bounded function $G^{(n-1)} \colon \mathbb{R}^{2(n-1)} \to \mathbb{C}$ obeying 
	\[ 
		{\rm supp} \;  G^{(n-1)} \subseteq \overline{(
\widetilde{V}^+ )^{(n-1)}} , 
	\]
such that $ \widetilde {\mathfrak W}_{C}^{(n-1)}=P(D)G^{(n-1)} $,  
with 
 	\[
 		P(D) = \frac{  \partial^{k_1 + \ldots +  k_{n-1} + l_1 + \ldots l_{n-1} } }
 		{ \partial E_1^{k_1}\partial p_1^{l_1} \cdots  \partial E_1^{k_{n-1}}\partial p_1^{l_{n-1}}} \, , 
		\qquad k_i, l_i \in \mathbb{N}\, . 
	\]
Consequently an analytic continuation ${\mathcal W}_+^{(n-1)}$
of ${\mathfrak W}_C^{(n-1)}$ to $(S_{\beta}\times\mathbb{R}-iV^{+})^{n-1}$ can be defined:
	\[ 
		\begin{array}{rl}
			& {\mathcal W}_+^{(n-1)} (\xi_1-i\eta_1, \ldots , \xi _{n-1}-i\eta_{n-1})  = 
			\\ [4mm]
			& \quad = (2 \pi \beta)^{-(n-1)} P \bigl(-i(\xi_1-i\eta_1, \ldots , \xi _{n-1}-i\eta_{n-1}) \bigr) \times 
			\\ [4mm]
			& \quad \qquad \times \int_{(S_{\beta}\times\mathbb{R})^{n-1}} \prod_{j=1}^{n-1}{\rm d} p_j {\rm d} E_j \; 
			{\rm e}^{-i(\xi_j -i \eta_j) \cdot (p_j, E_j)} \, 
			\times 
			\\ [4mm]
			& \quad \qquad \quad \qquad \quad \qquad \quad \qquad 
				\times G^{(n-1)} \big( (p_1, E_1), \ldots, (p_{n-1}, E_{n-1}) \big)\, . 
		\end {array}
	\] 
If $ \eta_j \in V^{+}$ for all $j \in \{1, \ldots , n-1 \}$, this integral exists. Furthermore, its boundary value
for $(\eta_1, \ldots , \eta_{n-1}) \searrow 0$ is ${\mathfrak W}_C^{(n-1)}$. Polynomial boundedness of the analytic 
function ${\mathcal W}_+^{(n-1)}$ results from the following inequality \cite[Theorem IX.16]{RS}:
	\[
		\begin{array}{rl}
			& \left| {\mathcal W}_+^{(n-1)}(\xi_1-i\eta_1, \ldots , \xi _{n-1}-i\eta_{n-1}) \right|
			\\ [4mm]
			& \qquad \leq C \; \left| P\bigl(-i(\xi_1-i\eta_1, \ldots , \xi _{n-1}-i\eta_{n-1}) \bigr) \right| \; 
			\left(1+d((\eta_1, \ldots, \eta_{n-1}))^{-N} \right) . 
		\end{array}
	\]
$C$ is a constant, $d((\eta_1, \ldots, \eta_{n-1}))$ is the distance of $(\eta_1, \ldots, \eta_{n-1})$ to $\partial (V^{+})^{n-1}$
and $N$ is a positive integer. 
% \qed
\end{proof}

\bigskip
Next we investigate the consequences of locality on the circle $S_\beta 
\equiv [0, \beta )$. 

\begin{lemma} \label{5.2} 
The tempered distributions $ {\mathscr W}_{C}^{(n)}(\alpha_1, \sigma_1, \ldots , \alpha_{n}, \sigma_{n}) $ defined in (\ref{wc})
are real valued for $( \alpha_1, \sigma_1, \ldots, \alpha_n, \sigma_{n}) \in J^{(n)}$, where 
	\begin{align}
		\label{defJ}
		(\alpha_1  , \sigma_1, \ldots, \alpha_n  , \sigma_n ) \in J^{(n)} 
		\Leftrightarrow \left\{ 
			\begin{array}{l}
			\; (\alpha_i  , \sigma_i) \in S_\beta \times \mathbb{R} ,   \\ 
			\; ( \alpha_{i+1} - \alpha_{i}, \sigma_{i+1}-\sigma_{i}) \in \lambda_i V_{\beta} , \\
			\; \sum_{i =1}^{n-1} \lambda_i = 1, \quad \lambda_i >0 ,\\ 
	\end{array}
	\right.
\end{align}
with $V_\beta:= \{ (\alpha, \sigma) \mid |\sigma| < \alpha < \beta - |\sigma| \} \subseteq S_{\beta} \times \mathbb{R} $ and $i=1, \ldots, n-1$. 
\end{lemma}

\begin{proof} Assume that the space-time points $(\alpha_i, \sigma_i)$ and $(\alpha_j, \sigma_j)$ are space-like to each other for all 
choices of $i \ne j$ and $i, j \in \{1, \ldots, n\}$. Then, as a consequence of locality, all the field operators $\phi_{C}(\alpha_i, \sigma_i)$ 
commute (as quadratic forms) with each other and therefore ${\mathscr W}_{C}^{(n)}( \alpha_1, \sigma_1, \ldots, \alpha_n , \sigma_{n})$ equals
	\[
		\begin{array}{rl}
			\bigl(\Omega_{C}, \phi_{C}(\alpha_1, \sigma_1)\cdots \phi_{C}(\alpha_n, \sigma_n) \Omega_{C}\bigr)  
			& 
			= \bigl(\Omega_{C}, \phi_{C}(\alpha_n, \sigma_n)\cdots \phi_{C}(\alpha_1, \sigma_1) \Omega_{C}\bigr) 
			\\ [4mm]
			& 
			= \overline{ {\mathscr W}_{C}^{(n)}(\alpha_1, \sigma_1, \ldots, \alpha_n , \sigma_{n} )}.
		\end{array}
	\]
In other words, the tempered distributions $ {\mathscr W}_{C}^{(n)}(\alpha_1, \sigma_1, \ldots ,\alpha_n, \sigma_{n}) $ are real valued. 
Thus the lemma follows, once we have shown that  the set $J^{(n)}$ consists of points, which are pairwise space-like 
to each other. 

A point $(\alpha,\sigma)$ on the cylinder is space-like to~the origin~$(0,0)$ iff $(\alpha,\sigma) \in V_{\beta}$. Space-likeness is a 
symmetric relation and therefore it suffices to prove that $(\alpha_i,\sigma_i)$ is space-like to~$(\alpha_j,\sigma_j)$ for $i>j$, \emph{i.e.},  
	\begin{equation} 
		\label{cfl}
		(\alpha_{j},\sigma_{j})-(\alpha_{i},\sigma_{i}) \in V_{\beta} \quad \textrm{ for } \quad i > j. 
	\end{equation}
Moreover, for $0< \lambda \le 1$,   
	\[ 
		V_{\lambda \beta}= \left\{(\alpha,\sigma) \in W  \mid |\alpha| + |\sigma|< \lambda \beta \right\} ,  
	\]
with $W$ the wedge $\{(\alpha,\sigma) \in [0,\beta)\times\mathbb{R} \mid \alpha > |\sigma| \}$. The map ${\mathfrak n} \colon  
[0,2 \pi) \times \mathbb{R} \rightarrow \mathbb{R}^{+}$,
	\[ 
		(\alpha,\sigma) \mapsto |\alpha| + |\sigma|,   
	\]
defines a norm. Denote its  restriction to the wedge $W$ by ${\mathfrak n}_{|W}$.  Equ.~(\ref{cfl}) now follows from the triangle inequality:
	\begin{align}
		{\mathfrak n}_{|W} \left( (\alpha_{j},\sigma_{j})-(\alpha_{i},\sigma_{i}) \right) &= {\mathfrak n}_{|W} \bigl((\alpha_{j}-\alpha_{j-1},\sigma_{j} - \sigma_{j-1}) + \ldots
		\nonumber \\
		&\qquad \qquad \qquad \ldots  +(\alpha_{i+1}-\alpha_{i},\sigma_{i+1} - \sigma_{i}) \bigr)\nonumber\\ 
		&\leq {\mathfrak n}_{|W}((\alpha_{j}-\alpha_{j-1},\sigma_{j} -\sigma_{j-1})) + \ldots 
		\nonumber \\
		&\qquad \qquad \qquad \ldots  + {\mathfrak n}_{|W}((\alpha_{i+1}-\alpha_{i},\sigma_{i+1} - \sigma_{i})) \nonumber\\ 
		&< \lambda_{j-1}\beta + \lambda_{j-2}\beta + \ldots  + \lambda_{i}\beta \leq
		\beta \sum_{k =1}^{n-1} \lambda_k = \beta,  \nonumber
	\end{align}
and therefore (\ref{defJ}) implies (\ref{cfl}). We note that the set of $n$ points on the cylinder, 
which are space-like to each other, is actually larger than $J^{(n)}$.  
% \qed
\end{proof}

Because the tempered distributions $ {\mathscr W}_{C}^{(n)}(\alpha_1, \sigma_1, \ldots , \alpha_{n}, \sigma_{n}) $ 
defined in (\ref{wc}) are real valued for $( \alpha_1, \sigma_1, \ldots, \alpha_n, \sigma_{n}) \in J^{(n)}$, we can 
apply the Schwarz reflection principle. The function  
\begin{eqnarray*}
%\label{36}
& {\mathcal W}_-^{(n-1)} (\xi_1+ i \eta_1, \ldots , \xi _{n-1} + i\eta_{n-1}) = 
\overline{{\mathcal W}_+^{(n-1)}(\xi_1-i\eta_1, \ldots , \xi _{n-1}-i\eta_{n-1})} 
\qquad 
\quad 
\\
&
={\displaystyle \int_{(S_{\beta}\times \mathbb{R})^{n-1}} \kern - .3cm
\frac{\prod_{j=1}^{n-1} {\rm d} p_j {\rm d} E_j}{(2 \pi \beta)^{n-1} } \; }
{\rm e}^{ i(\xi_j + i \eta_j) \cdot (p_j, E_j)} \,  
\overline{\widetilde {\mathfrak W}_{C}^{(n-1)} \big( (p_1, E_1), \ldots, (p_{n-1}, E_{n-1}) \big)} . 
\nonumber
\end{eqnarray*}
is analytic on 
$(S_{\beta}\times \mathbb{R} +i V^+) \times \cdots \times (S_{\beta}\times \mathbb{R} +i V^+) $ 
and polynomially bounded as $\eta_i \searrow 0$. Since $V^{+}$ is a cone, $V^{+}\times \ldots \times V^{+}$ 
is  a cone (by definition). Applying the Edge-of-the-Wedge theorem \cite[Theorem 2-16]{SW}, we conclude 
that there exists a complex neighbourhood ${\mathcal N}$ of $\lambda_1 V_{\beta} \times \ldots \times \lambda_{n-1} 
V_{\beta}$ and a function ${\mathcal W}_C^{(n-1)}$  defined and holomorphic  in 
${\mathcal N} \, \cup \,(S_{\beta}\times \mathbb{R} -iV^{+})^{n-1} \cup (S_{\beta}\times  \mathbb{R} +iV^{+})^{n-1} $, 
which coincides  with the restriction of the distributions ${\mathfrak W}_{C}^{(n-1)}( \xi_1, \xi_2, \ldots, \xi_{n-1})$ 
defined in (\ref{37}) to $\lambda_1 V_{\beta} \times \ldots \times \lambda_{n-1} V_{\beta}$. In fact, by only 
partially reordering the fields (see the proof of  Lemma~\ref{5.2})  and using the support properties of the 
Fourier transform stated in Theorem \ref{5.1}, we can extend ${\mathcal W}_C^{(n-1)}$ into the regions 
$(S_{\beta}\times \mathbb{R}  \mp i V^+) \times \cdots \times (S_{\beta}\times \mathbb{R} \mp i V^+) $ 
(the $\mp$ all being independent).
Note that relative coordinates are used in ${\mathcal W}_C^{(n-1)}$ and therefore reordering of the arguments 
results in $i V^+$ being replaced by $-iV^+$. 
Thus we arrive at the following result:  

\begin{theorem} 
\label{Th5}
There exists of a function ${\mathcal W}_C^{(n-1)}$  holomorphic in 
	\begin{equation}\label{c-n-1}
		{\mathcal C}^{(n-1)}  := {\mathcal N} \, \cup \,{\mathcal D}^{(n-1)} , 
	\end{equation}
which coincides  with the restriction of the 
distributions ${\mathfrak W}_{C}^{(n-1)}( \xi_1, \xi_2, \ldots, \xi_{n-1})$ defined in (\ref{37})
to $\lambda_1 V_{\beta} \times \ldots \times \lambda_{n-1} V_{\beta}$.
Here ${\mathcal N}$ is a complex neighbourhood of 
$\lambda_1 V_{\beta} \times \ldots \times \lambda_{n-1} V_{\beta}$
and 
	\[
		{\mathcal D}^{(n-1)} := 
		\left( \lambda_1 V_{\beta} \times \ldots \times \lambda_{n-1} V_{\beta}\right)\, 
		+ \,  i \, (V^{+} \cup V^{-}) \times \ldots \times (V^{+} \cup V^{-}) \, 
	\]
with $\lambda_i >0$ and $\sum_{i=1}^{n-1} \lambda_i =1$. (In fact, one
can take the union over these $\lambda_j$'s, $j= 1, \ldots, n-1$). 
\end{theorem}

% \begin{remark} 
% In case $n=2$, the holomorphic hull 
%	\[
%		H ({\mathcal D}^{(1)}) = \mathbb{C^2} \setminus \overline{ \bigcup_{(a,b, \lambda) 
%		\in N(V_\beta)}  \{ z \in \mathbb{C}^2 \mid (z_0- a)^2 - (z_1-b)^2 = \lambda^2 \}} 
%	\]
% of the set ${\mathcal D}^{(1)}$ was explicitly 
% computed by Bros, Messiah and Stora \cite{BMS}. 
% Here the set of admissible hyperboloids for the set $V_\beta$ is 
%	\[
%		N(V_\beta) \doteq \{ (a, b,  \lambda) \in \mathbb{R}^2 \times \mathbb{R}^+ 
%		\mid (\alpha- a)^2 - (\sigma-b)^2 < \lambda^2  \; \; \forall (\alpha,\sigma) \in V_\beta \}  \; . 
%	\]
% We note that this set does \underline{not} contain the set $\lambda V_{\beta} + i \mathbb{R}^2 $.  
% \end{remark}

\section{The relativistic KMS condition for the $P (\phi)_2$ model}
\label{sectrelkms}

In the previous section we have seen that the Wightman functions on the circle are the boundary 
values of a function ${\mathcal W}_C^{(n-1)}$  holomorphic in the region ${\mathcal C}^{(n-1)}$ 
(see (\ref{c-n-1})). Now define a new function 
	\begin{equation}
		{\mathcal W}_{\beta}^{(n-1)}:={\mathcal W}_{C}^{(n-1)}\circ \Xi^{-1} \; , 
		\label{nel}
	\end{equation}
where $\Xi$ is 
the coordinate transformation 
	\[
		(z_1,w_1,\ldots,z_{n-1},w_{n-1}) \mapsto  (iz_1,-i w_1,\ldots,i z_{n-1},-i w_{n-1})
	\]
on $\mathbb{C}^{2(n-1)}$. Then ${\mathcal W}_\beta^{(n-1)}$ is 
analytic in the domain 
	\begin{equation}
	\label{initialdomain}
		\bigl((Q^- \cup Q^+)^{n-1} - i \lambda_1 V_{\beta} \times \ldots \times \lambda_{n-1} V_{\beta} \bigr)\cup  \Xi\mathcal{N} \; , 
	\end{equation}
with $\lambda_i >0$ and $\sum_{i=1}^{n-1} \lambda_i =1$, 
where the right and left wedges are 
	\[ 
		Q^\pm = \left\{ (\tau, y) \in \mathbb{R}^2 \mid \pm y > |\tau| \right\} . 
	\]
Our aim is to show that
\begin{itemize}
\item[$i.)$] 
the thermal Wightman functions $ { {\mathcal W}}_{\beta}^{(n-1)}$ introduced in~(\ref{nel}) 
extend to functions analytic in the product of domains 
	\begin{equation}
		\label{strongtube}
		(\lambda_1 {\cal T}_{\beta}) \times \cdots \times (\lambda_{n-1} {\cal T}_{\beta}), 
		\qquad  {\cal T}_{\beta} := \mathbb{R}^2 - i V_{\beta} , \qquad \sum_{j=1}^{n-1} \lambda_j = 1,
	\end{equation}
and $\lambda_j >0$, $j=1, \ldots, n-1$. In fact, one can take the union over these~$\lambda_j$'s;
\item[$ii.)$] 
the boundary values of the analytic functions ${\mathcal W}_{\beta}^{(n-1)}$ as $\Im z_j \searrow 0$ yield 
tempered distributions. 
\end{itemize}
We will also ensure that these tempered distributions are indeed the Wightman distributions of the 
thermal field theory on the real line. We proceed in several steps.

\subsection{Products of sharp-time fields and their domains}
\label{sec:3.2}
 
The representation $\pi_\beta$ defined in Section \ref{sectosrec} is a regular CCR representation 
(see \cite{GeJ2}), and therefore one can define for $h \in C^\infty_{0 \, \mathbb{R}} (\mathbb{R})$ 
the Segal field operators 
	\begin{equation}
		\label{stone}
		\phi_\beta  (h) := - i \frac{{\rm d}}{ {\rm d} s} \pi_\beta \left({\rm e}^{i \phi(0,s h)} \right)\Bigl|_{ s=0} \;  . 
	\end{equation}
While Stone's theorem is convenient to show that  $\phi_\beta  (h)$ exists as a self-adjoint unbounded
operator, it provides little control on the domain of $\phi_\beta  (h)$. In fact, a priori it is not even clear 
whether $\Omega_\beta$ is an element of ${\mathcal D}(\phi_\beta  (h))$. We will need several steps 
to resolve these domain problems. 

\goodbreak
\begin{lemma} \quad
\begin{itemize}
	\item[$ i.)$] {\rm (Products of sharp-time fields)}.
	\label{pos}
	Let $h_i\in  {\mathscr S}_{ \mathbb{R}}^{\infty}(\mathbb{R})$  for $i =1, \ldots, j$, $j \in \mathbb{N}$, 
	and~$0 \le \alpha_1 \le \ldots \le \alpha_{j}<\beta$. Then
		\begin{equation}
			\label{euclsharl2}
			\phi (\alpha_j, h_j) \cdots \phi (\alpha_1, h_1)  \in
			\bigcap_{1\leq p<\infty}L^{p}(Q, \Sigma_{[0,\alpha_j]}, {\rm d} \mu), \quad j \in \mathbb{N} \, .
	\end{equation}
\item[$ ii.)$] {\rm (Convergence of sharp-time Schwinger functions, Part II)}.
Let $h_{i}\in C^{\infty}_{0\: \mathbb{R}}(\mathbb{R})$ and $\alpha_{i}\in S_{\beta}$, $1\leq i\leq n$. Then
	\[
		\lim_{l\to \infty}\int_{Q} \phi (\alpha_n, h_n) \cdots \phi (\alpha_1, h_1) \; {\rm d}\mu_{l}= \int_{Q}
		\phi (\alpha_n, h_n) \cdots \phi (\alpha_1, h_1)\; {\rm d}\mu.
	\]
\end{itemize}
\end{lemma}

\begin{proof} $i.)$ Consider an approximation of the Dirac $\delta$-function: $\delta_{\kappa}(x):= 
\kappa \chi(\kappa x)$, with $\chi$  a function in $C_0^\infty(\mathbb{R})$ and $\int \chi(x)\, {\rm d} x=1$. 
It has been shown in \cite[Proposition 7.3]{GeJ2} that
	\[
		\lim_{k \to \infty} \;    \phi \left( \delta_k ( \, .\, - \alpha_i) \otimes   h_i  \right)  \in
		\bigcap_{1\leq p<\infty}L^{p}(Q, \Sigma, {\rm d} \mu), 
		\qquad h_i\in  {\mathscr S}_{ \mathbb{R}} (\mathbb{R}). 
	\]
For later purpose, we briefly recall the proof:
	\begin{align*}
	& \int_Q \bigl( \phi \left(   \delta_k ( \, .\,  - \alpha_i) \otimes   h_i \right) \bigr)^{p} \; 
			{\rm d} \mu
	\\ 
	& \qquad = (-i)^{p} \frac{{\rm d}^{p} }{{\rm d} \lambda^{p}} 
			\left( \Omega_C \, , \, W_{[ - \infty, + \infty]} 
			\left(\lambda \bigl(  \delta_k ( \, .\,  - \alpha_i) \otimes   h_i \bigr)\right) 
			\Omega_C \right) \Bigl|_{ \lambda =0} \; , 
	\end{align*}
where $W_{[ a, b]} (f)$ is a solution of the heat equation 
	\[ 
		\frac{{\rm d}}{{\rm d} b} W_{[ a, b]} (f)= W_{[ a, b]} (f) \bigl(- H_C + i \phi_C (f_b) \bigr) ,  \qquad a \le b, 
	\]
with the boundary condition $W_{[a,a]}(f)= \mathbb{1}$ 
and with $f_b (\, . \, ) := f (\, . \, , b) \in {\mathscr S}_\mathbb{R} (S_\beta )$ for
$f  \in {\mathscr S}_\mathbb{R} (S_\beta \times \mathbb{R})$. Now, 
if $f = \delta_k ( .- \alpha_i) \otimes   h_i$, then the function
$f_x \in {\mathscr S}_\mathbb{R}(S_\beta)$ 
is equal to $\delta_k (\, . \, - \alpha_i)  h_i (x)$. It follows from (\ref{Fockcontinuity2}), \emph{i.e.}, estimate 
(5.9) in \cite[Proposition 5.4]{GeJ2}, that  $h_i\in  C_{0 \,\mathbb{R}}^{\infty}(\mathbb{R})$ implies
	\begin{align*}
			\pm \phi_C \left( \delta_k ( \, .\,  - \alpha_i)    h_i (x)  \right)  
			& \le c \left\|  \delta_k ( \, .\,  - \alpha_i) \,   h_i (x)   \right\|_{H^{-1}(S_\beta)}  (H_C +1) 
			\\ 
			& \le c \, | h_i (x) |   \left\|   \delta_k  \right\|_{H^{-1}(S_\beta)}    (H_C +1) . 
	\end{align*}
Set $r_k(x):= c \, | h_i (x) | \left\|   \delta_k  \right\|_{H^{-1}(S_\beta)}  $ and apply  \cite[Lemma A.8]{GeJ2} to obtain 
	\begin{equation}
		\label{cb1}
		\left\|\frac{{\rm d}^{p}}{{\rm d} \lambda^{p}} W_{[ - \infty, + \infty]} 
		\left(\lambda \left(  \delta_k ( \, .\,  - \alpha_i) \otimes   h_i \right)\right) \right\| \le  p! \; \|  r_k \|^{p}_\infty \,
		{\rm e}^{ \| r_k \|_1 \| r_k \|_\infty^{-1}} . 
	\end{equation} 
Since $\delta_k (\, .\, - \alpha_i )$ converges to $\delta (\, .\, - \alpha_i )$ in $H^{-1}(S_\beta)$ and
$h_i\in  C_{0 \,\mathbb{R}}^{\infty}(\mathbb{R})$ for $i= 1, \ldots,  j$, we see that $\lim_{k \to \infty} \| r_k \|_1 < \infty$ 
and $\lim_{k \to \infty}\| r_k \|_\infty < \infty$. Thus 
	\begin{equation}
		\label{cb2}
		\int_Q  \left| \phi (\alpha_i, h_i) \right|^j \; {\rm d} \mu     < \infty \, .  
	\end{equation}
The estimate \eqref{euclsharl2} with $p=1$ 
then follows from the H\"older inequality  
	\[
		\int_Q  | \phi (\alpha_j, h_j) \cdots \phi (\alpha_1, h_1) |\; {\rm d} \mu
		\le \prod_{i=1}^j \left( \int_Q  \left| \phi (\alpha_i, h_i) \right|^j \; {\rm d} \mu  \right)^{1/j} \; . 
	\]
The higher $L^p$-estimates follow as well, as $(\alpha_k, h_k)$ may 
equal $(\alpha_l, h_l)$, $k, l =1, \ldots, j$.
$\Sigma_{[0,\alpha]}$-measurability follows from the fact that $a.)$ for all $k$
there is an~$\epsilon_k$ (the $\delta_k$ were chosen to have compact support) such that 
	\[
		\phi \left(  \delta_k ( \, .\, - \alpha_i) \otimes   h_i \right)  \in
		\bigcap_{1\leq p<\infty}L^{p}(Q, \Sigma_{[0,\alpha_i+\epsilon_k]}, {\rm d} \mu)
	\]
and $b.)$ the upper continuity of $\mu$. 

$ii.)$ Now let $h_{i}\in C^{\infty}_{0\: \mathbb{R}}(\mathbb{R})$ and $\alpha_{i}\in S_{\beta}$, 
$1\leq i\leq n$. Part $ii.)$ follows from  
	\begin{align*}
		&\lim_{l\to \infty}\int_{Q} \phi (\alpha_n, h_n) \cdots \phi (\alpha_1, h_1) \; {\rm d}\mu_{l}
		\nonumber 
		\\
		& \qquad =  \lim_{k\to \infty} \frac{{\rm d}^{n}}{{\rm d} \lambda_1 \cdots {\rm d} \lambda_n }  
		 \left( \Omega_C \, ,  W_{[ - a, a]} \left( {\textstyle\sum_{i =1}^n\lambda_i 
		 \left(  \delta_k ( \, .\,  - \alpha_i) \otimes   h_i \right) } \right)
		\Omega_C \right) \Bigl|_{\lambda_i = 0} 
		\nonumber 
		\\
		& \qquad = \int_{Q} \phi (\alpha_n, h_n) \cdots \phi (\alpha_1, h_1)\; {\rm d}\mu \; , 
	\end{align*} 
for ${\rm supp}\;  \delta_k ( \, .\,  - \alpha_i) \otimes   h_i  \subset S_\beta \times [-a,a] $, $i= 1, \ldots, n$, 
as for $s \le -a \le a \le t $ the map $(s, t) \mapsto (\Omega_C , W_{[s, t]} (f) \Omega_C)$ is constant. 
% \qed
\end{proof}

The existence of products of sharp time fields in $L^2 (Q, \Sigma, {\rm d} \mu)$ allows us to investigate 
their domains, taking advantage of their Euclidean heritage:  

\goodbreak
\begin{proposition} 
\label{neu}
Let $h_{i}\in C^{\infty}_{0\: \mathbb{R}}(\mathbb{R})$, $1\leq i\leq n$. Then 
\begin{itemize}
\item	 [$i.)$] $\Omega_\beta \in {\mathcal D} (L)$ and $ L \Omega_\beta = 0$; 
\item [$ii.)$] If $\alpha_1, \ldots, \alpha_n  \ge 0$ and $ \sum_{j= 1}^n \alpha_j \le \beta / 2$, then 
	\begin{equation}
		\label{kl}
		{\rm e}^{- \alpha_{n-1} L } \phi_\beta ( h_{n-1}) \cdots {\rm e}^{- \alpha_{1}  L } 
		\phi_\beta ( h_{1})   \Omega_\beta \in {\mathcal D} \bigl( \phi_\beta (h_{n}) \bigr)
	\end{equation}
and 
	\begin{equation}
		\label{kln} 
		\phi_\beta (h_{n}) {\rm e}^{- \alpha_{n-1} L } \phi_\beta ( h_{n-1}) \cdots {\rm e}^{- \alpha_{1}  L } 
		\phi_\beta ( h_{1})   \Omega_\beta \in {\mathcal D} \bigl({\rm e}^{- \alpha_{n} L }
		\bigr) . 
	\end{equation}
Moreover,  the linear span of such vectors is dense in ${\mathcal H}_\beta$ and
	\begin{align*} 
		&{\rm e}^{-\alpha_n L } \phi_\beta (h_{n}) {\rm e}^{-  \alpha_{n-1} L } 
		\phi_\beta ( h_{n-1}) \ldots {\rm e}^{- \alpha_{1}  L } \phi_\beta (  h_{1}) \Omega_\beta 
		\\
		& \qquad \qquad = {\cal V}\Bigl(U(\alpha_n,0) \phi(0,h_{n}) U(\alpha_{n-1},0)\phi(0,h_{n-1})
		\cdots U(\alpha_1,0)  \phi(0,h_{1}) \Bigr)  ;
	\end{align*} 
\item [$iii.)$] If $0\le \alpha_{1}\leq \cdots\leq \alpha_k \le \beta /2$ and $\beta/2 \le \alpha_{k+1}\le \ldots \le \alpha_{n}\leq \beta$, 
then 
	\begin{align} 
		\label{31}
		&  \int_{Q} \Bigl( \prod_{j=1}^{n} \phi(\alpha_{j}, h_{j} ) \Bigr) \; {\rm d}\mu 
		\nonumber \\  
		&  \quad = 
		\bigl( {\rm e}^{(\alpha_n-\beta) L } \phi_\beta (h_{n}) {\rm e}^{(\alpha_{n-1}-\alpha_{n})  L } \phi_\beta (h_{n-1}) 
		\cdots  {\rm e}^{(\alpha_{k+1}-\alpha_{k+2}) L } \phi_\beta (h_{k+1}) \Omega_\beta \; ,  \; \quad
		\nonumber \\  
		& \qquad \qquad   {\rm e}^{- \alpha_1 L } \phi_\beta (h_{1}) {\rm e}^{(\alpha_1 -  \alpha_2) L } 
		\phi_\beta (h_{2})  \cdots  {\rm e}^{(\alpha_{k-1}-\alpha_{k}) L } \phi_\beta (h_{k}) \Omega_\beta \bigr)   \, .
	\end{align} 
\item [$iv.)$]   $ \|  {\rm e}^{- (\beta/2) L}  \phi_\beta (h_{n}) \cdots  \phi_\beta ( h_{1}) \Omega_\beta \bigr  \| = \|  \phi_\beta (h_{n})     
		\cdots  \phi_\beta ( h_{1})  \Omega_\beta \bigr  \|$. 
\end{itemize}
\end{proposition}

\begin{proof} 
Let us first note that (\ref{31}) formally results from differentiating the following identity, which is a conequence of 
${\rm e}^{i \phi(0,h_{j})} \in L^{\infty}(Q, \Sigma_{\{0\}}, {\rm d} \mu)$ for $h_{i}\in C^{\infty}_{0\: \mathbb{R}}(\mathbb{R})$ 
and the Osterwalder-Schrader reconstruction outlined in Section 2.3:
	\begin{align*}
		& \int_{Q} \Bigl( \prod_{j=1}^{n}{\rm e}^{i \phi(\alpha_{j}, h_{j})} \Bigr) {\rm d}\mu   
			\nonumber \\ 
		&
		\quad = \int_{Q} R \biggl( \overline{  U(\beta, 0) \prod_{j=k+1}^{n}  
		{\rm e}^{-i\phi(-\alpha_{j}, h_{j} )} } \biggr) 
				\prod_{j=1}^{k} {\rm e}^{i\phi(\alpha_{j}, h_{j} ) }\; {\rm d}\mu 
		\nonumber\\ 
		& \quad = 
			\Bigl( {\cal V} \bigl( U(\beta, 0)  {\rm e}^{-i \phi(-\alpha_n,h_{n})}  
			\ldots  {\rm e}^{-i \phi(-\alpha_{k+1},h_{k+1})} \bigr)  \; ,
			{\cal V} ( {\rm e}^{i \phi(\alpha_k ,h_{k})}   \cdots {\rm e}^{i \phi(\alpha_1,h_{1})} )  \Bigr)  
		\nonumber \\ 
		&\quad = \bigl( {\rm e}^{(\alpha_n-\beta) L } {\rm e}^{-i\phi_\beta (h_{n}) } 
		{\rm e}^{(\alpha_{n-1}-\alpha_{n})  L } {\rm e}^{-i\phi_\beta (h_{n-1}) }
			\cdots  {\rm e}^{(\alpha_{k+1}-\alpha_{k+2}) L } {\rm e}^{-i\phi_\beta (h_{k+1}) }
			\Omega_\beta \; ,  \; 
		\nonumber \\ 
		&\quad \quad  \qquad \qquad  {\rm e}^{- \alpha_1 L } {\rm e}^{i\phi_\beta (h_{1})} 
		{\rm e}^{(\alpha_1 -  \alpha_2) L } 
			{\rm e}^{i\phi_\beta (h_{2})  }\cdots  {\rm e}^{(\alpha_{k-1}-\alpha_{k}) L } 	
			{\rm e}^{i \phi_\beta (h_{k}) } \Omega_\beta \bigr)   \, ,  
		\label{cweyl}
	\end{align*} 
for $1\leq i\leq n$, and $0\le \alpha_{1}\leq \cdots\leq \alpha_k \le \beta /2$ and $\beta/2 \le 
\alpha_{k+1}\le \ldots \le \alpha_{n}\leq \beta$.  
Note that inserting the identity $U(\beta, 0) = \mathbb{1}$ ensures that $(\beta - \alpha_i) \in 
[0, \beta/2 ]$
for $i= k+1, \ldots, n$.
However, we have to ensure that  (\ref{31}) is well-defined.

\begin{itemize} 
\item [$i.)$] See  \cite[Lemma 8.4]{KL2}: $ 1 \in {\mathcal M}_\alpha$, thus 
$\Omega_\beta \in {\mathcal D}_\alpha$ and $ {\rm e}^{- \alpha L} \Omega_\beta 
= P(\alpha) \Omega_\beta = \Omega_\beta$ as $U(\alpha, 0)1 = 1$ for $0 \le \alpha \le \beta $;
\item [$ii.)$]
The case $n=1$, namely $ \Omega_\beta \in {\mathcal D} \left( \phi_{\beta}(h_{1}) \right)$ and 
	\[  
		{\rm e}^{- \alpha_1 L} \phi_{\beta}(h_1)\Omega_\beta \in {\mathcal H}_\beta 
		\quad \hbox{for} \quad 0\le \alpha_1 \le \beta /2  
	\]
was proven in \cite{GeJ3}. In fact, 
	\[  
		{\rm e}^{- \alpha_1 L} \phi_{\beta}(h_1)\Omega_\beta \in {\mathcal D} (\phi_{\beta}(h_2))  \, ,  
	\]
as $\phi(0,h_2)$ acts as a multiplication operator on $\phi (\alpha_1, h_1)$ and 
	\[ 
		\phi(0,h_2) \phi (\alpha_1, h_1) \in {\mathcal M}_{\beta/2-\alpha_1}   
	\]
by Lemma \ref{pos} $i.)$. As $ P(\alpha) {\mathcal D}_\gamma \subset {\mathcal D}_{\gamma- \alpha}$, 
it follows that 
	\[  
		{\rm e}^{- \alpha_2 L} \phi_{\beta}(h_2){\rm e}^{- \alpha_1 L} \phi_{\beta}(h_1)\Omega_\beta \in {\mathcal D}_{\beta/2-\alpha_1- \alpha_2}    
	\]
and $ \phi (0, h_3) \phi (\alpha_2, h_2) \phi (\alpha_1 + \alpha_2, h_1)  \in {\mathcal M}_{\beta/2-\alpha_1 - \alpha_2} $ implies 
	\[  
		{\rm e}^{- \alpha_2 L}\phi_{\beta}(h_2){\rm e}^{- \alpha_1 L} \phi_{\beta}(h_1)
		\Omega_\beta \in {\mathcal D} ( \phi_{\beta}(h_3) ) \, .   
	\]
Iterating this argument it follows that   
	\[ 
		{\cal V} \bigl( \phi (\alpha_k, h_k) \cdots \phi (\alpha_1 + \ldots + \alpha_k, h_1) \bigr) 
		\in {\mathcal D}_{\beta/2-\gamma}, 
	\] 
if $\sum_{i=1}^{k} \alpha_k \le \gamma \le \beta/2$. Thus (\ref{kl}) and (\ref{kln}) follow. 

Next we prove that 
	\[ 
		{\rm e}^{-\alpha_n L } \phi_\beta (h_{n}) {\rm e}^{-  \alpha_{n-1} L } \phi_\beta ( h_{n-1})  
		\ldots {\rm e}^{- \alpha_{1}  L } \phi_\beta (  h_{1}) \Omega_\beta
	\] 
is  dense in ${\mathcal H}_\beta$ for $\alpha_1, \ldots, \alpha_n  \ge 0$ and $\sum_{j= 1}^n \alpha_j \le \beta / 2$. Assume that, 
for $\Psi \in {\mathcal H}_{\beta}$ and all $f, g \in C^{\infty}_{0\: \mathbb{R}}(\mathbb{R})$,
	\begin{equation}
		\label{32}  
			\forall m,n \in \mathbb{N}: \qquad( \Psi , \phi_\beta (f)^n {\rm e}^{-  \beta L / 2} 
			\phi_\beta ( g)^m  \Omega_\beta ) = 0 . 
	\end{equation}
(Note that (\ref{32}) is well-defined as a consequence of (\ref{kln}).) 
Then\footnote{From Prop.~A6 $i.)$ and Theorem 7.2 $i)$ in \cite{GeJ2} it 
follows that the 
vector valued function 
$s, t \mapsto  {\cal V} ( {\rm e}^{i s \phi( 0, f) } 
			{\rm e}^{i t \phi (\beta/2, g) } ) $ is entire for 
			$f, g \in C^{\infty}_{0\: \mathbb{R}}(\mathbb{R})$.} 
	\begin{equation} 
		\label{ortho}
			(\Psi \, , \, {\rm e}^{i \phi_\beta (f) } {\rm e}^{-\beta L /2} 
			{\rm e}^{i \phi_\beta (g) } \Omega_\beta) = 0 \, . 
	\end{equation} 
But vectors of the form 
${\rm e}^{i \phi_\beta (f) }  {\rm e}^{-\beta L /2} {\rm e}^{i \phi_\beta (g) } 
\Omega_\beta$, 
$f, g \in C^{\infty}_{0\: \mathbb{R}}(\mathbb{R})$,
are dense \cite[Theorem 11.2]{KL1} in ${\mathcal H}_\beta$, and therefore (\ref{ortho}) implies 
$\Psi = 0$, establishing the claim. 

\item [$iii.)$] If $0\le \alpha_{1}\leq \ldots\leq \alpha_k \le \beta /2$ and $\beta/2 \le \alpha_{k+1}\le \ldots \le \alpha_{n}\leq \beta$, 
then according to $ii.)$
	\begin{align*}
			&\bigl( {\rm e}^{(\alpha_n-\beta) L } \phi_\beta (h_{n}) {\rm e}^{(\alpha_{n-1}- \alpha_n)  L } \phi_\beta (h_{n-1}) 
				\ldots  {\rm e}^{(\alpha_{k+1}-\alpha_{k+2}) L } \phi_\beta (h_{k+1}) \Omega_\beta \; ,  \;
			\\ 
			& \qquad \qquad \qquad {\rm e}^{- \alpha_1 L } \phi_\beta (h_{1}) {\rm e}^{- (\alpha_2 -  \alpha_1) L } 
				\phi_\beta (h_{2})  \ldots  {\rm e}^{- (\alpha_{k} -\alpha_{k-1}) L } \phi_\beta (h_{k}) \Omega_\beta \bigr)
	\end{align*}
is well-defined and equals 
	\[
		\begin{array}{rl}
			&  \Bigl( {\cal V} \bigl(  \phi(\beta-\alpha_n,h_{n})  \ldots  \phi(\beta -\alpha_{k+1},h_{k+1}) \bigr)  \; ,
				{\cal V} \bigl( \phi(\alpha_k ,h_{k}) \cdots \phi(\alpha_1,h_{1}) \bigr)  \Bigr)  
			=
			\\ [4mm]
			& \qquad = \int_{Q} R \left( \overline{  \prod_{j=k+1}^{n} \phi(\beta-\alpha_{j}, h_{j} ) } \right)
				\prod_{j=1}^{k} \phi(\alpha_{j}, h_{j} ) \; {\rm d}\mu 
			\\ [4mm]
			& \qquad = \int_{Q} R  \left( \overline{U(\beta, 0) \prod_{j=k+1}^{n} \phi(-\alpha_{j}, h_{j} ) } \right)
				\prod_{j=1}^{k} \phi(\alpha_{j}, h_{j} ) \; {\rm d}\mu 
			\\ [4mm]
			& \qquad = \int_{Q} R \left(  \prod_{j=k+1}^{n} \phi(-\alpha_{j}, h_{j} ) \right) 
				\prod_{j=1}^{k} \phi(\alpha_{j}, h_{j} ) \; {\rm d}\mu 
			\\ [4mm]
			& \qquad = \int_{Q}   \left( \prod_{j=k+1}^{n} \phi(\alpha_{j}, h_{j} ) \right) 
				\prod_{j=1}^{k} \phi(\alpha_{j}, h_{j} ) \; {\rm d}\mu 
			\\ [4mm]
				& \qquad = \int_{Q}\prod_{j=1}^{n} \phi(\alpha_{j}, h_{j} ) \; {\rm d}\mu \, .
		\end{array} 
	\]
We made again use of $U(\beta, 0) =\mathbb{1}$, which holds by periodicity.

\item [$iv.)$]  By $ii.)$ we have $ \phi_\beta (h_{n})   \phi_\beta ( h_{n-1}) \cdots  \phi_\beta ( h_{1})   
\Omega_\beta \in {\mathcal D} \bigl({\rm e}^{- \beta L / 2} \bigr)  $. Now 
	\[  
		\begin{array}{rl} 
			& \left\| {\rm e}^{- \beta  L / 2}  \phi_\beta (h_{n})   \phi_\beta ( h_{n-1}) \cdots  \phi_\beta ( h_{1})   \Omega_\beta\right\|^2 
			=
			\\[4mm]
			& \quad =  \left\| {\cal V} \bigl( U (\beta/2, 0) \phi ( 0, h_{n}) \cdots  \phi ( 0, h_{1}) \bigr) \right\|^2 
			\\[4mm]
			& \quad =  
				\int_{Q} \overline{ U (\beta/2, 0) \phi ( 0, h_{n}) \cdots  \phi ( 0, h_{1})} \; R \, U (\beta/2, 0) \; \phi ( 0, h_{n}) 
				\cdots  \phi ( 0, h_{1}) \; {\rm d}\mu 
			\\[4mm]
			& \quad =  
				\int_{Q} \overline{ \phi ( 0, h_{n}) \cdots  \phi ( 0, h_{1}) } \; U(-\beta/2, 0) \, R \, U (\beta/2, 0) \; \phi ( 0, h_{n}) 
				\cdots  \phi ( 0, h_{1}) \; {\rm d}\mu
			\\[4mm]
			& \quad =  
				\int_{Q}   \phi ( 0, h_{n}) \cdots  \phi ( 0, h_{1}) \; R \, U (\beta, 0) \; \phi ( 0, h_{n}) \cdots  \phi ( 0, h_{1}) \; {\rm d}\mu
			\\[4mm]
			& \quad =    
			\left\| {\cal V} \left(  \phi ( 0, h_{n}) \cdots  \phi ( 0, h_{1}) \right) \right\|^2 
			=   \| \phi_\beta (h_{n})   \phi_\beta ( h_{n-1}) \cdots  \phi_\beta ( h_{1})   \Omega_\beta \|^2 \; , 
			\end{array}
	\]
again using $U(\beta, 0) = \mathbb{1}$.
% \qed
\end{itemize}
\end{proof}

The extension of these results to real times is our next objective. Given the self-adjoint operator $\phi_{\beta}(h)$,
$h \in C_{0 \, \mathbb{R}}^{\infty}(\mathbb{R})$, set 
	\[ 
		\phi_{\beta}(t, h):= {\rm e}^{i t L}\phi_{\beta}(h) {\rm e}^{- i t L} , \qquad t \in \mathbb{R} \, .  
	\]
The domain of the self-adjoint operator $\phi_{\beta}(t, h)$ is ${\rm e}^{it L} {\mathcal D} ( \phi_{\beta}( h))$. 
That products of field operators   
smeared out in time can be applied to the distinguished vector~$\Omega_\beta$ will be shown in the final subsection.

\subsection{Analyticity properties of the thermal Wightman distributions}
\label{sectwight}
We can now proceed by using the following remarkable consequence of the KMS condition established by Araki (see Equ. (1.27) in Lemma A, 
\cite{AM}).

\begin{lemma}[Araki] 
\label{L7}
Let $\omega_\beta$ be a $(\tau,\beta)$-KMS state over a von Neumann algebra~${\mathcal R}$.
Let $(z_1, \ldots, z_{n-1})\in \mathbb{C}^{n-1}$ with $ \Im z_j  \ge 0$ for 
$j=1 , \ldots , n-1$, and
	\begin{align*}
		\Im z_1 & + \ldots + \Im z_{k-1} + \Im z_{k}' \le \beta/2 \; , \qquad \Im z_{k}' \ge 0 \; , \\
		\Im z_{n-1} &+ \ldots + \Im z_{k+1} + \Im z_{k}'' \le \beta/2 \; , \qquad \Im z_{k}'' \ge 0 \; , \quad z_{k}'+ z_{k}'' = z_{k} \; . 
	\end{align*}
Moreover, let $z_n= i - \sum_{i=1}^{n-1} z_i$. It follows that there exists some $j= 1, \ldots, k$ such that 
for $A_0, \ldots A_n \in {\mathcal R}$ one has
	\begin{align*}
		& \bigl( {\rm e}^{\overline{ i z_k'}  L} A_{k}^* {\rm e}^{\overline{i z_{k-1}}  L} A_{k-1}^* \cdots
		{\rm e}^{ \overline{i z_1} L } A_1^*  \Omega_\beta 
%		\\ 
%		&  
		, {\rm e}^{iz_{k}''  L} A_{k+1} {\rm e}^{iz_{k+1}  L} A_{k+2} \cdots
		{\rm e}^{i z_{n-1}  L } A_n {\rm e}^{i z_n L} \Omega_\beta  \bigr) \\ 
		& \qquad = \bigl( {\rm e}^{\overline{ i z_{k+1}'}  L} A_{k+1}^* {\rm e}^{\overline{ i z_{k}}  L} A_{k}^* \cdots
		{\rm e}^{\overline{ i z_{j+1}} L } A_{j+1}^*  \Omega_\beta \, , 
		\\ 
		&  
		\qquad \qquad   \quad  {\rm e}^{iz_{k+1}''  L} A_{k+2} {\rm e}^{iz_{k+2}  L} A_{k+3} \cdots
		{\rm e}^{i z_{n-1}  L } A_n {\rm e}^{i  z_n  L}  A_{1} \cdots {\rm e}^{i  z_{j-1}  L}  A_{j}
		{\rm e}^{i z_{j}  L}    
		\Omega_\beta  \bigr),
	\end{align*}
with  $\Im z_{k+1}' , \Im z_{k+1}'' \ge 0$, 
	\begin{align*}
		\Im z_j & + \Im z_{j+1} + \ldots + \Im z_{k} + \Im z_{k+1}'' \le \beta/2 \; , \\
		\Im z_{j-1} & + \Im z_{j-2} + \ldots + \Im z_{1} + \Im z_{n} + \ldots + \Im z_{k+2} + \Im z_{k+1}' \le \beta/2 \; ,   \;    
	\end{align*}
and $z_{k+1}' + z_{k+1}''= z_{k+1}$ for some $k = 0, 1, \ldots , n-1$. 
\end{lemma}

The identity stated applies to bounded operators of the fields, but the fields themselves may be 
approximated by bounded operators (see, \emph{e.g.}, Equ.~\eqref{cutoff-field} below). 
After removing these approximations, one finds 
	\begin{align*} 
				& \bigl( {\rm e}^{\overline{ i z_k''}  L} \phi(h_{k}) {\rm e}^{\overline{ i z_{k-1}}  L} \phi(h_{k-1}) \cdots
		{\rm e}^{\overline{i z_1} L } \phi(h_1)  \Omega_\beta \, , \, 
		\\ 
		&  
		\qquad \qquad \qquad \qquad \qquad  {\rm e}^{iz_{k}'  L} \phi(h_{k+1}) {\rm e}^{iz_{k+1}  L} \phi(h_{k+2}) \cdots
		{\rm e}^{i z_{n-1}  L } \phi(h_n) {\rm e}^{ i z_n  L}\Omega_\beta  \bigr) \\ 
		& = \int_{\mathbb{R}^{2}} {\rm d} x_{1} \cdots {\rm d} x_{n} \;  h_{1} (x_1) \cdots h_{n} (x_n ) 
%		\\
%		& \quad \times 
		\mathcal{W}_{\beta}^{(n-1)}(z_1, x_1 - x_2, 
		\ldots , z_{n-1}, x_{n-1} - x_n) .
	\end{align*}
Note that for $\Re z_1 = \ldots = \Re z_n = 0$ the existence of the l.h.s.~follows from Proposition \ref{neu} $ii.)$.
The extension to non-vanishing real parts will be discussed below. But before we do so, we 
choose sequences of absolutely integrable functions 
$h^{(k)}_{i}\in C_{0}^{\infty}(S_{\beta})$, $i=1, \ldots , n$,  tending to the Dirac distributions $ \delta ( \, . \, - x_k)$ as $k \to \infty$.
For $\Re z_i = 0 $ and $\Im z_i >0 $, $i=1, \ldots , n-1$, the limit $k \to \infty$ exists and yields 
	\begin{align*} 
		& \mathcal{W}_{\beta}^{(n-1)}(z_1, x_1 - x_2, 
		\ldots , z_{n-1}, x_{n-1} - x_n)
		\\
		& = \bigl( {\rm e}^{\overline{i z_k''}  L + i (x_{k}- x_{k+1}) P } \phi(\delta) 
		{\rm e}^{\overline{i z_{k-1}}  L + i (x_{k-1}- x_{k})P } \phi(\delta) \cdots
		{\rm e}^{\overline{i z_1} L + i (x_{1}-x_2)P} \phi(\delta) \Omega_\beta \, , \, 
		\\ 
		&  
		\qquad \qquad \quad  
%		{\rm e}^{ i (x_{k}- x_{k+1}) P }
		{\rm e}^{iz_{k}'  L  } \phi(\delta) {\rm e}^{iz_{k+1}  L - i (x_{k+1}- x_{k+2}) P } \phi(\delta) \cdots
		{\rm e}^{i z_{n-1}  L - i (x_{n-1}- x_n)P } \phi(\delta) \Omega_\beta  \bigr)  .
	\end{align*}
Setting $x_1' = x_1 - x_2$, $x_2' = x_2 - x_3$, etc., this identity takes the following form
	\begin{align*} 
		& \mathcal{W}_{\beta}^{(n-1)}(z_1, x_1', 
		\ldots , z_{n-1}, x_{n-1}')
		\\
		& = \bigl( {\rm e}^{\overline{i z_k''}  L + i x_k''P } \phi(\delta) 
		{\rm e}^{\overline{i z_{k-1}}  L + i x_{k-1}' P } \phi(\delta) \cdots
		{\rm e}^{\overline{i z_1} L + i x_{1}' P} \phi(\delta) \Omega_\beta \, , \, 
		\\ 
		&  
		\qquad \qquad \qquad \qquad  {\rm e}^{iz_{k}'  L -  i x_k'P  } \phi(\delta) {\rm e}^{iz_{k+1}  L - i x_{k+1}' P } \phi(\delta) \cdots
		{\rm e}^{i z_{n-1}  L - i x_{n-1}' P } \phi(\delta) \Omega_\beta  \bigr)  .
	\end{align*}
We have set $x_k = x_k'+x_k''$, using the same ratio of the absolute values as in the splitting of $z_k = z_k'+z_k''$. 

\goodbreak
In particular, 
	\begin{align*} 
		& \|  {\rm e}^{iz_{k}'  L -  i x_k' P  } \phi(\delta) {\rm e}^{iz_{k+1}  L - i x_{k+1} P } \phi(\delta) \cdots
		{\rm e}^{i z_{2k-1}  L - i x_{2k-1} P } \phi(\delta) \Omega_\beta  \bigr) \|^2
		\\
		& = \mathcal{W}_{\beta}^{(2k-1)}(z_{2k-1}, x_{2k-1}, \ldots, z_{k+1}, x_{k+1},
		z_{k}, x_{k},  z_{k+1}, x_{k+1} \ldots , z_{2k-1}, x_{2k-1}) \;  
	\end{align*}
with $z_k= 2 z_k' $ and $x_k = 2 x_k'$. It follows that the vector valued function 
	\begin{align*}
	 & (t_k, x_k, \ldots  t_{n-1},  x_{n-1}) \mapsto \\
	&  \qquad 
	 {\rm e}^{it_{k}  L -  i x_kP  } \phi(\delta) {\rm e}^{it_{k+1}  L - i x_{k+1} P } \phi(\delta) \cdots
		{\rm e}^{i t_{2k-1}  L - i x_{2k-1} P } \phi(\delta) \Omega_\beta
	\end{align*}
can be analytically continued into the region 
	\begin{equation}
	\label{vectordomain}
		 (Q^- \cup Q^+)^{k} - i \lambda_k V_{\beta} \times \ldots \times \lambda_{2k-1} V_{\beta}    \; , 
	\end{equation}
with $\lambda_i >0$ and $\sum_{i=k}^{2k-1} \lambda_i =1/2$. 

On the other hand, 
by first applying Lemma \ref{L7} and then removing the approximations in a similar manner as above, one can establish the following identity:
	\begin{align} 
		\label{new-k-W}
		& \mathcal{W}_{\beta}^{(n-1)}(z_1, x_1, 
		\ldots , z_{n-1}, x_{n-1})
		\nonumber \\
		& = \bigl( {\rm e}^{\overline{i z_{k+1}'}L + i x_{k+1}' P  } \phi (\delta) {\rm e}^{\overline{i z_{k}}  L + i  x_k' P } \phi(\delta) \cdots
		{\rm e}^{\overline{i z_{j+1}} L + i x_{j+1}'  P } \phi(\delta) \Omega_\beta \, , \,
		 \\ 
		&  
		\;   {\rm e}^{iz_{k+1}''  L } \phi (\delta)
		 {\rm e}^{iz_{k+2}  L- i x_{k+2} P}   \cdots 	
		 {\rm e}^{i z_n  L - i x_{n} P } \phi (\delta) {\rm e}^{i  z_{1}  L- i x_{1} P }  \cdots {\rm e}^{i  z_{j-1}  L- i x_{j-1} P }  \phi (\delta)
		\Omega_\beta  \bigr)  .
		\nonumber
	\end{align}
where $x_n=  - x_{n-1}  - x_{n-2} - \ldots - x_1$. 

Clearly, there are $n-1$ different expressions for $\mathcal{W}_{\beta}^{(n-1)}$ which can be gained by 
repeated application of Lemma \ref{L7}.
		
\begin{theorem} \label{wightanal} 
The thermal Wightman functions $ { {\mathcal W}}_{\beta}^{(n-1)}$ introduced in~(\ref{nel}) are 
analytic in the product of domains 
	\begin{equation}
		\label{strongtube}
		(\lambda_1 {\cal T}_{\beta}) \times \cdots \times (\lambda_{n-1} {\cal T}_{\beta}), 
		\qquad  {\cal T}_{\beta} := \mathbb{R}^2 - i V_{\beta} , \qquad \sum_{j=1}^{n-1} \lambda_j = 1,
	\end{equation}
and $\lambda_j >0$, $j=1, \ldots, n-1$.  In fact, one can take the union over these $\lambda_j$'s. 
\end{theorem} 

\begin{proof} 
We recall that $\mathcal{W}_{\beta}^{(n-1)}$ is an analytic function in the domain \eqref{initialdomain}. 
Within the domain \eqref{initialdomain}, the Cauchy Schwarz inequality yields 
	\begin{align*} 
		& \left| \mathcal{W}_{\beta}^{(n-1)} (z_1, x_1 + i y_1, 
		\ldots , z_{n-1}, x_{n-1} + i y_{n-1} ) \right| \\
		& \quad  \le  \left\| {\rm e}^{\overline{i z_{k}'} L + y_{k}'P} \phi_{\beta}(\delta_0) {\rm e}^{\overline{i z_{k-1}} L + i (x_{k-1}- iy_{k-1})P}
		\cdots {\rm e}^{\overline{i z_{1}} L + i (x_{1}- iy_{1})P} \phi_{\beta}(\delta_0) \Omega_{\beta } \right\|  
		\\
		& \qquad   \times \left\|
		{\rm e}^{i z_{k}'' L - i (x_{k}+ iy_{k}'')P} \phi_{\beta}(\delta_0) \cdots {\rm e}^{i z_{n-1} L - i (x_{n-1}+ iy_{n-1})P}
		\phi_{\beta}(\delta_0)\Omega_{\beta} \right\| .
	\end{align*}
Here $y_k = y_k' + y_k''$ is split according to the same ratio as $z_k = z_k' + z_k''$.

As $L$ and $P$ are self-adjoint operators, the spectral theorem implies that 
the vector valued function 
	\[
		 (z_{k}'' , w_{k}'') \mapsto 
		 {\rm e}^{i z_{k}'' L - i w_{k}'' P} \phi_{\beta}(\delta_0) \cdots {\rm e}^{i z_{n-1} L - i (x_{n-1}+ iy_{n-1})P}
		\phi_{\beta}(\delta_0)\Omega_{\beta}  
	\]
is analytic in the domain $(z_{k}'' , w_{k}'') \in \mathbb{R}^2 + i \frac{\lambda_k}{2} V_\beta$ (as the norm of the vector is preserved 
by applying the unitary  ${\rm e}^{i \Re z_{k}'' L - i \Re w_{k}''P}$). 
And consequentely, the function $\mathcal{W}_{\beta}^{(n-1)}$ extends to an analytic function in the 
domain  
	\begin{align}
	\label{new-k-domain}
		& \Bigl( (Q^- \cup Q^+) - i \lambda_1 V_{\beta}  \Bigr) \times \ldots \times \Bigl( (Q^- \cup Q^+) - i \lambda_{k-1} V_{\beta}  \Bigr) 
		  \\
		& \; \;  \times
		\Bigl( \mathbb{R}^2 - i \lambda_k V_{\beta}  \Bigr) \times
		\Bigl( (Q^- \cup Q^+) - i \lambda_{k+1} V_{\beta}  \Bigr) \times \ldots \times \Bigl( (Q^- \cup Q^+) - i \lambda_{n-1} V_{\beta}  \Bigr),
		\nonumber
	\end{align}
with $\lambda_i >0$ and $\sum_{i=1}^{n-1} \lambda_i =1$. Using the expression \eqref{new-k-W}
one can replace $k$ by $k+1$ in \eqref{new-k-domain}.
Iterating this procedure (eventually renaming the variables) and applying Hartogs' theorem \cite[p.~30]{V}
one concludes that ${\mathcal W}_{\beta}^{(n-1)}$ is analytic in the domain \eqref{strongtube}.
% \qed
\end{proof}

\subsection{Temperedness of the thermal Wightman distributions}
\label{sectwight}

For a quantum system confined to a  box, Ruelle~\cite{Rue1}\cite{Rue2} has used a
H\"older inequality, which applies to the trace in the Gibbs state. It was pointed out by Fr\"ohlich \cite{Fr1} that
this H\"older inequality is crucial in the present context.  

\begin{theorem}[H\"older inequality]
\label{hoelder} Let $\omega_\beta$ be a $(\tau,\beta)$-KMS state over a von Neumann algebra ${\mathcal R}$.
Define for, $p \in \mathbb{N}$  and $A \in {\mathcal R}^+$, 
	\[
		\| A \|_{p} :=  \omega_\beta \bigl(\underbrace{ {\rm e}^{it L/p } A \cdots  
		{\rm e}^{it L/p }A }_{p \textrm{ times}} \bigr)_{\upharpoonright t=i \beta}^{1/p} \; .
	\]
Let $(z_1, \ldots, z_n)\in \mathbb{C}^n$ with $0 \le \Re z_j$, $\sum_{j=1}^m \Re z_j \le 1/2 $ and $\sum_{j=m+1}^n \Re z_j
\le 1/2 $, and let $p_j$ be the smallest, positive integer such that
$\frac{1}{p_j}\le \min\{\Re z_{j+1},\Re z_j\}$, with $z_{n+1}=z_n$ and $z_0=z_1$. Then 
	\begin{align}
		\label{crucialbound}
		\Bigl| \omega_\beta \bigl( A_n  {\rm e}^{it_n \beta L } \cdots A_1{\rm e}^{it_1 \beta L}
			A_0 \bigr)_{\upharpoonright t_j=iz_j} \Bigr| 
		\le \| A_0 \|_{p_0} \cdots \| A_n \|_{p_n}  
	\end{align}
for all $A_0, \ldots, A_n \in {\mathcal R}^+$. (The subscript $\upharpoonright t_j=iz_j$ indicates the analytic continuation 
from $t_j$ to $iz_j$, $j=1, \ldots, n$.)
\end{theorem}

\begin{proof} The proof of this results relies on the theory of non-commutative $L^p$-spaces and is given in 
\cite{RJ}. 
\end{proof}

Note that because of the time-invariance of the KMS state, the r.h.s.~in (\ref{crucialbound}) does not depend on 
$\Im z_i$, $i= 1, \ldots, n$. 

\begin{proposition} 
\label{prop2}
For  $0\le \epsilon \le 1$ fixed there exist constants $c_1, c_2 >0$ such that
	\begin{equation}
		\label{new-phi-bound}
		\pm \phi_C (g) \le  c_1 \, \| g \|_{H^{  - \frac{1}{2}-\frac{\epsilon}{2} } (S_{\beta})} (H_C + c_2)^{\frac{1}{2}+\epsilon}   
	\end{equation}
for all $g \in H^{  - \frac{1}{2}-\frac{\epsilon}{2}  }(S_{\beta})$.
\end{proposition}

\begin{proof} Set $H_0 = {\rm d} \Gamma (\nu)$. It is sufficient to prove that  
	\[
		A (g) \equiv (H_0+ 1)^{-\frac{1}{4} -  \frac{\epsilon}{2} } \phi_C (\nu^{\frac{1}{2} 
		+ \frac{\epsilon}{2} } g) (H_0+ 1)^{-\frac{1}{4} -  \frac{\epsilon}{2}}   
	\]
is a bounded operator on Fock space, uniformly bounded for $\|g \|_2 \le 1$. The first order estimate (see, \emph{e.g.},  \cite[Equ. (2.21)]{ros})
	\[
		(H_0 +1)   \le c_3   (H_C +c_2)   \quad \textrm{for } c_2, c_3 \gg 1
	\]
and operator monotonicity of the map 
$\lambda \mapsto \lambda^\alpha$ for $0\le \alpha \le 1$ (see, \emph{e.g.}, \cite[Example 4.6.46]{KR}) then ensure 
the fractional $\phi$-bound (\ref{new-phi-bound}).

We now follow ideas of Rosen (see, \emph{e.g.},  \cite[Proof of Lemma 6.2]{ros}). We show that $A(g)$ is a bounded bilinear form in Fock space. 
The desired operator extension then follows from the Riesz representation theorem. It is sufficient to show that, for $\|g \|_2\le 1$,
	\[
		| ( \Phi, A(g) \Psi ) | \le c_4 \| \Phi \| \cdot \| \Psi \| \, , 
	\]
for $\Phi, \Psi$ arbitrary vectors on Fock space and $c_4 > 0$ a constant. Now (see (\ref{creat-ann}))
	\begin{align*}
		| ( \Phi, A(g) \Psi ) | 
		&\le{1 \over \sqrt{2}} \Bigl( | ( (H_0+ 1)^{-\frac{1}{4} -  \frac{\epsilon}{2} } \Phi, 
			a^{*}  (\nu^{\frac{\epsilon}{2} } g) (H_0+ 1)^{-\frac{1}{4} -  \frac{\epsilon}{2}}  \Psi ) |
		\nonumber \\
		& \qquad \quad+| ( a^{*}  (\nu^{\frac{\epsilon}{2} } \bar g)(H_0+ 1)^{-\frac{1}{4} -  \frac{\epsilon}{2} } \Phi, 
		(H_0+ 1)^{-\frac{1}{4} -  \frac{\epsilon}{2}} \Psi ) |  \Bigr) \; . 
	\end{align*}
Since $H_C^{(0)}$ commutes with the number operator, and both terms are of the same structure, it is sufficient to prove that
for $\Phi_n \in \mathcal{H}_C^{(n)} $ and $\Psi_{n-1} \in \mathcal{H}_C^{(n-1)} $  with $\| \Phi_n \|, \| \Psi_{n-1} \|\le 1$ 
	\begin{align*}
		&| ( (H_C^{(0)}+ 1)^{-\frac{1}{4} -  \frac{\epsilon}{2} } \Phi_n , a^{*}  (\nu^{\frac{\epsilon}{2}} g) 
		(H_C^{(0)}+ 1)^{-\frac{1}{4} -  \frac{\epsilon}{2}} \Psi_{n-1} ) |
		\nonumber \\
		&\qquad \qquad \le \|  (n+ 1)^{-\frac{1}{4} } \Phi_n \|   \cdot  
			\| a^{*}  (\nu^{\frac{\epsilon}{2}} g) (H_C^{(0)}+ 1)^{-\frac{1}{4} - \frac{\epsilon}{2}}  \Psi_{n-1}    \|   
		\label{eq:fockcalc}
	\end{align*}
is uniformly bounded in $n$.
For simplicity it is assumed 
(in the second inequality below) 
that the mass $m\ge 1$, so that 
$\sum_{i=1}^{n-1} \nu (k_i) + 1\ge n$; 
otherwise one is left with yet another $n$-independent constant. Now 
	\begin{align*}
		&  (n+ 1)^{-1/2} \, \|\Phi_n\|^2 \, \|  a^{*}  (\nu^{\frac{\epsilon}{2} } g) 
		(H_C^{(0)}+ 1)^{-\frac{1}{4} - \frac{\epsilon}{2}} \Psi_{n-1}    \|^2
		\\
		&\quad   
		\le \frac{n}{(n+1)^{1/2}} \, \|\Phi_n\|^2 
		\\
		&
		\qquad \quad
		 \times \frac{1}{\beta^n}\int \prod_{j=1}^n \frac{{\rm d} k_j}{\nu(k_j)}    \Biggl| 
                \frac{\nu (k_n)^{ \frac{\epsilon}{2}}  \, \tilde g(k_n) }
                {\left(\sum_{i=1}^{n-1} \nu (k_i) + 1 \right)^{1/4+\epsilon/2}} 
		\,  \widetilde \Psi_{n-1}  (k_1, k_2, \ldots k_{n-1})    \Biggr|^2
		\\
		&\quad  \le  
		\frac{n}{(n+1)^{1/2}n^{1/2} } 
		\, \|\Phi_n\|^2 
		\\
		&
		\qquad \quad
		\times 
		\frac{1}{\beta^n} 
		\int \prod_{j=1}^n \frac{{\rm d} k_j}{\nu(k_j)}  
		 \Biggl| 
		\left( \frac{\nu (k_n)  }{\sum_{i=1}^{n-1} \nu (k_i) + 1} 
		\right)^{ \frac{\epsilon}{2}}
		\tilde g		(k_n)\,  
		\widetilde \Psi_{n-1} 
		 (k_1, k_2, \ldots k_{n-1})   \Biggr|^2 \nonumber
		\\
		&\quad   \le \| g \|_2^2 \, \|\Phi_n\|^2 \, \| \Psi_{n-1} \|^2 \;,  
%		\label{eq:fockcalc2}
	\end{align*}
which establishes the claim.   
% \qed
\end{proof}

The Euclidean time zero field $ \phi (0, h) \in L^p (Q, \Sigma, {\rm d} \mu)$, $1 \le p < \infty$, can be approximated 
by a sequence of  functions in $L^\infty (Q, \Sigma, {\rm d} \mu)$. The latter can be decomposed in their positive and negative 
part. Define, for $h \in {\mathscr S}_{\mathbb{R}}(\mathbb{R})$ and $\alpha \in [0, \beta)$,  
        \begin{align}
        \phi_\pm^{(\ell)} (\alpha, h) & = \begin{cases} \pm \phi (0, h) & \text{if} \quad 0 \le  \pm \phi (\alpha, h)  \le \ell \\
        0 & \text{otherwise} \; .
        \end{cases} 
        \end{align}
It follows from \cite[Lemma~3.5]{Seiler-Simon} that $\phi_+^{(\ell)} (0, h) - \phi_-^{(\ell)} (0, h)$ converges to $ \phi (0, h)$ as
$\ell \to \infty$ in any $L_q$-norm with $q <p$. We can use this result to define approximations for the thermal time-zero 
field $\phi_\beta(h)$: set 
        \begin{align}
	\label{cutoff-field}
	\phi_\ell^\pm (h) := \pi_\beta \bigl( \phi_\pm^{(\ell)} (0, h) \bigr)
        \end{align}
and $\phi_\ell^\pm (t, h) :={\rm e}^{i t L} \phi_\ell^\pm (h) {\rm e}^{- i t L} $, $t \in \mathbb{R}$. 

\begin{lemma} 
\label{Lm7new}
For $h \in {\mathscr S}_{\mathbb{R}}(\mathbb{R})$ and $p \in \mathbb{N}$ even,
the expressions
\begin{equation}
\label{newbound}
\begin{array}{rl}
& |\kern -.5mm |\kern -.5mm | \, h  |\kern -.5mm |\kern -.5mm |_{p }  
 :=  \max \,  \Bigl\{   
 \lim_{\ell \to \infty} \| \phi^+_\ell (t, h) \|_p \, , \, \lim_{\ell \to \infty}  
 \| \phi^-_\ell (t, h) \|_p  \Bigr\}
\end{array} 
\end{equation}
are bounded  from above by $ \sqrt[p]{p!}\cdot |h|_{\cal S} $, for some Schwartz norm $|\, . \, |_{\cal S} $. 
%\label{frohcoroll}
\end{lemma}

\begin{proof}
Set $\phi_\pm (\alpha,h)= \lim_{\ell \to \infty}\phi_\pm^{(\ell)} (\alpha, h)$, $\alpha \in [0, \beta)$. 
We have 
	\begin{align} 
		\label{eq:hnormhoelder}
		|\kern -.5mm |\kern -.5mm | \, h  |\kern -.5mm |\kern -.5mm |_{p }^p  
			&
			=  \max    \left\{ \lim_{\ell \to \infty} \int_{Q} \prod_{k=1}^{p } \phi_{+}^{(\ell)} \!  
			\left( \frac{k\beta}{p}, h \right) \, {\rm d}\mu \, , \, 
						\lim_{\ell \to \infty} \int_{Q} \prod_{k=1}^{p } \phi_{-}^{(\ell)} \!  
						\left( \frac{k\beta}{p}, h \right) \, {\rm d}\mu \right\} 
			\nonumber \\
			&=  \max    \left\{ \int_{Q} \prod_{k=1}^{p } \phi_{+} \!  
			\left( \frac{k\beta}{p}, h \right) \, {\rm d}\mu \, , \, 
						\int_{Q} \prod_{k=1}^{p } \phi_{-} \!  
						\left( \frac{k\beta}{p}, h \right) \, {\rm d}\mu \right\} 
			\nonumber \\
			&\le \max  \left\{  \int_Q  \phi_{+} \left(0,  h \right)^{p} {\rm d} \mu \, \, , \, 
			\int_Q  \phi_{-} \left(0,  h \right)^{p} {\rm d} \mu \right\} .
	\end{align}
In the first line we have used the definition of the norm $\| \cdot \|_p \,  $, in the third line we have used 
the H\"older inequality for $L^p(Q,\Sigma,{\rm d}\mu)$ and  translation invariance of ${\rm d} \mu$.
  
Since $p\in\mathbb{N}$ is even and  the supports of $\phi_{+}(0,h)$ and $\phi_{-}(0,h)$ are disjoint, 
	\begin{align}
%		\label{EQ78}
		%|\kern -.5mm |\kern -.5mm | \, h  |\kern -.5mm |\kern -.5mm |_{p }^p 
		\int_Q  \phi_{\pm} \left(0,  h \right)^{p} {\rm d} \mu&\le \int_{{\rm supp} \,  \phi_{+}(0,h) }
		\phi_{+} (0,h)^p \, {\rm d} \mu +\int_{{\rm supp} \,  \phi (0,h)_{-}} \phi_{-} (0,h)^p \, 
		{\rm d} \mu\nonumber\\
		&= \int_Q\phi(0,h)^p \, {\rm d} \mu \, . \label{eq:phipolardec}
	\end{align}
Taking advantage of the fractional $\phi$-bound (\ref{new-phi-bound}) we can 
apply~\cite[Lemma A.7, p.~167]{GeJ2}, which states that there is a constant $c >0$ such that
	\begin{equation} \label{limitineq}
		\left| \int {\rm d} \mu \; {\rm e}^{ i \lambda \phi \left( \delta_{k} \otimes h \right) } \right|
		\le {\rm e}^{ c  |\Im \lambda|^\varkappa \| r_{k} \|_\varkappa^\varkappa}
		\quad 
	\end{equation}
for $ \lambda \in \mathbb{C}$, $ \varkappa = (\frac{1}{2}-\epsilon)^{-1}$, $0 < \epsilon < 1/2$ and
	\[
		r_k(x):=  \, | h (x) | \left\|   \delta_k  \right\|_{H^{- \frac{1}{2} 
		- \frac{ \epsilon}{2}} (S_{\beta})}  \; .
	\]
The limit $k \to \infty$ exists, as the Dirac $\delta$-function is in all Sobolev spaces 
$H^{q}$ for $q < -1/2$. Denote $r (x) := \lim_{ k \to \infty}r_k(x)$. Applying Cauchy's 
formula\footnote{As mentioned before, Prop.~A6 $i.)$ and Theorem 7.2 $i)$ in \cite{GeJ2} imply that the 
map $\lambda \mapsto  \int {\rm d} \mu \; {\rm e}^{ i \lambda \phi \left( \delta_{k} \otimes h \right)}$
is entire.} on the circle of radius $R$ centred around 
$\lambda = 0$ 
yields
	\[
		\int_{Q} \phi\! \left( 0, h \right)^p \, {\rm d}\mu  \le p! \, R^{-p} 
		{\rm e}^{ c  R^\varkappa \| r \|_\varkappa^\varkappa}  \; . 
 	\]
Optimizing this bound w.r.t.~$R$ yields 
	\begin{align}
		\label{EQ82}
		\int_Q  \phi \left(  0,    h   \right)^{p}  {\rm d} \mu
		&\le   p ! \; \left(  \frac{c  \varkappa {\rm e}}{p}\right)^{p/ \varkappa} 
		\|  r \|_\varkappa^p \le p ! \, |h|_{\mathscr S}^p\;  . 
	 \end{align}
Here we have used $\sup_{p \in \mathbb{N}} \left(  \frac{c \varkappa {\rm e}}{p}\right)^{p/ \varkappa} < \infty$ 
and the fact that 
	\[
		\|  r \|_\varkappa =  \left\|   \delta  \right\|_{H^{- \frac{1}{2} - \frac{ \epsilon}{2}} (S_{\beta})}
		\left(\int | h (x) |^\varkappa {\rm d} x \right)^{1/\varkappa}
	\]
can be estimated by a Schwartz semi-norm if $h \in {\mathscr S}(\mathbb{R})$ and $\varkappa >2$.
Combining~(\ref{eq:phipolardec}) and (\ref{EQ82}) we arrive at  
	\[
		|\kern -.5mm |\kern -.5mm | \, h  |\kern -.5mm |\kern -.5mm |_{p }^p  
		\le  p ! \; |h|_{\mathscr S}^p \; , 
	\]
which establishes the lemma.% \qed
\end{proof}

\begin{remark} Using the $\phi$-bound (\ref{Fockcontinuity1}) one can use the equation preceding 
\cite[Equ.~(A.9), p.~165]{GeJ2} to arrive at  Fr\"ohlich's bound 
	\[
		\int {\rm e}^{\pm \phi (g\otimes h) } {\rm d} \mu \le {\rm e}^{c\int_{\mathbb{R}} {\rm d} x \, | h(x) |^2  C_\beta (g, g)  } \; ,  
		\quad g  \in  H^{- 1/2} ( S_\beta) \; , \; \; h \in  L_{\mathbb{R}}^2(\mathbb{R}) \; ,   
	\]
stated (for the special case $g = \mathbb{1}_{[0, l]}$ a characteristic function) in \cite[Equ.~(7)]{Fr1}.
However, using only this bound, we were unable to establish the existence of the products estimated 
in Lemma \ref{Lm7new}.
\end{remark}

\begin{theorem}  
\label{Th7}
The thermal Wightman functions
	\[
		{\mathfrak W}_\beta^{(n-1)} \bigl( t_1 - t_2, x_1 -x_2, \ldots, t_{n-1} - t_{n}, x_{n-1} - x_{n} \bigr)
	\]
are tempered distributions, which 
satisfy the \emph{relativistic KMS condition}, \emph{i.e.}, they
\begin{itemize}
\item[$i.)$] are the boundary values of  functions ${\mathcal W}_\beta^{(n-1)}$ analytic  in 
the interior of the product of domains 
	\begin{equation}
	\label{thermaltubes}
		(\lambda_1 {\cal T}_{\beta}) \times \cdots \times (\lambda_{n-1} {\cal T}_{\beta}), 
		\qquad {\cal T}_{\beta} := \mathbb{R}^2 - i V_{\beta}, 
	\end{equation}
where $\lambda_i >0$, $i= 1, \ldots, n-1$ and $\sum_{i=1}^{n-1} \lambda_i = 1$; 
\item[$ii.)$] satisfy the following boundary condition:  for any time-like vector $e = (e_0, e_1) \in V^+$ 
	\begin{align}
		\label{kms2-1}
		& \lim_{e \to 0} {\mathcal W}_\beta^{(n-1)} 
		\bigl( s_1, y_1, \ldots, s_{k-1}, y_{k-1} , s_k -   i e_0, 
			y_{k} - i e_1,  \ldots, s_{n-1} , y_{n-1} \bigr) \nonumber \\  
		& \qquad   = {\mathfrak W}_\beta^{(n-1)} \bigl( 
		s_1, y_1, \ldots, s_{n-1},  y_{n-1} \bigr) 
	\end{align}
and	\begin{align}
		\label{kms2-1}
		& \lim_{e \to 0} {\mathcal W}_\beta^{(n-1)} 
		\bigl( s_1, y_1, \ldots, s_{k-1}, y_{k-1} , s_k -  i \beta + i e_0, 
			y_{k} + i e_1,  \ldots, s_{n-1} , y_{n-1} \bigr) \nonumber \\  
		& \qquad   = {\mathfrak W}_\beta^{(n-1)} \bigl( s_k, y_k ,\ldots, s_{n},  y_{n}, 
		s_1, y_1, \ldots, s_{k-2}, y_{k-2} \bigr) 
	\end{align}
for all $(s_1, y_1\ldots, s_{n-1}, y_{n-1} )\in \mathbb{R}^{2(n-1)}$. We have 
set $s_k = t_{k} -t_{k+1}$ and $y_k = x_{k} -x_{k+1}$, $1 \le k < n$, and in addition,  $s_n = t_n - t_1$ and $y_n = x_n- x_1$. 
\end{itemize}
\end{theorem}

\begin{proof} 
The domain of analyticity of the thermal Wightman functions $ {\mathcal W}_{\beta}^{(n-1)}$  
stated in $i.)$ was established in Theorem \ref{wightanal}. Thus it remains to establish $ii.)$. 
We first prove that the boundary values 
in the distinguished time direction $(1,0)$  define
tempered distributions. In the sequel, we show that the boundary values 
in a time-like direction $e= (e_0, e_1)$ coincide with them. 

Within their domain of analyticity the Wightman functions can be approximated by 
the expectation values of bounded operators: let $h \in C_{0 \mathbb{R}}^\infty(\mathbb{R})$ 
and set 
	\[
		\phi_\ell (t, h) := \phi_\ell^+ (t, h)  - \phi_\ell^- (t, h) \, , \qquad t \in \mathbb{R}\, . 
	\]
To ease the notation, put 
	\begin{align*}
	{\underline s} & = (t_1 - t_2, \ldots, t_{n-1}-t_n) \; , \\
	{\underline \alpha} & = (\alpha_1- \alpha_2, \ldots, \alpha_{n-1}- \alpha_n) \; , \\
	{\underline y} & = (x_1 - x_2, \ldots, x_{n-1}-x_n) \; . 
	\end{align*}
Now define, using the nuclear theorem, the kernels ${\mathbf W}_{\ell_1, \ldots, \ell_n}
({\underline s} - i {\underline \alpha}, {\underline y})$ with $0 < \alpha_n < \ldots < \alpha_1 < \beta$,  
by requiring that
	\begin{align*}
		& \int {\rm d}x_1 \cdots {\rm d}x_n \; {\mathbf W}_{\ell_1, \ldots, \ell_n}
		({\underline s} - i {\underline \alpha}, {\underline y}) \; h_1(x_1) \cdots h_n(x_n) \\
		& \qquad \doteq  
		\omega_\beta \bigl(  \phi_{\ell_1} (r_1, h_{1})   \cdots 
		\phi_{\ell_n} (r_n, h_{n})   \bigr)_{\upharpoonright r_i=t_i + i \alpha_i}  
	\end{align*}
for all $h_1, \ldots, h_n  \in C_{0 \mathbb{R}}^\infty (\mathbb{R})$. 
As before, the subscript $\upharpoonright r_i=t_i + i \alpha_i$ indicates the analytic continuation 
from $t_i$ to $t_i + i \alpha_i$, $i=1, \ldots, n$.

Clearly,
\begin{align*}
		& \lim_{\ell_i \to \infty} {\mathbf W}_{\ell_1, \ldots, \ell_n}
		(-i {\underline \alpha}, {\underline y})  \nonumber \\
		& \qquad = 	
		\lim_{\ell_i \to \infty} \int_{Q} \phi^{(\ell_1)} (\alpha_1, x_1)  \cdots  \phi^{(\ell_n)} (\alpha_1, x_1)\, {\rm d}\mu
		\nonumber \\		
		& \qquad = 	
		{\mathcal W}_\beta^{(n-1)} 
		\bigl( - i (\alpha_1-\alpha_2), y_1, \ldots, 
		  - i (\alpha_{n-1} -\alpha_{n}), y_{n-1} \bigr) \; . 
\end{align*}
We have used the notation introduced in paragraph $ii.)$, Theorem \ref{Th7}. In addition, we 
have set $ \phi^{(\ell)} (\alpha, h) :=  \phi_+^{(\ell)} (\alpha, h) -  \phi_-^{(\ell)} (\alpha, h)$.
Since the functions involved are all bounded on compact sets of their domain 
of analyticity, it follows that for $0 < \alpha_n < \ldots < \alpha_1 < \beta$,  $i= 1, \ldots, n$,
\begin{align}
		& \lim_{\ell_i \to \infty} {\mathbf W}_{\ell_1, \ldots, \ell_n}
		({\underline s} - i {\underline \alpha}, {\underline y})  \nonumber \\
		& \qquad = 		
		{\mathcal W}_\beta^{(n-1)} 
		\bigl( s_1 - i (\alpha_1-\alpha_2), y_1, \ldots, 
		 s_{n-1} - i (\alpha_{n-1} -\alpha_{n}), y_{n-1} \bigr) \; , 
\end{align}
uniformly on compact sets in their domain of analyticity.  We denote this limit by 
${\mathbf W}({\underline s} - i {\underline \alpha}, {\underline y})$.

We now show that there exist uniform bounds 
(independent of $\ell_i$,$i =1, \ldots, n$) 
as we approach the real boundary of the domain of analyticity: by construction
	\begin{align}
		\label{72new}
		\int & {\rm d}x_1 \cdots {\rm d}x_n \; {\mathbf W}
		({\underline s} - i {\underline \alpha}, {\underline y}) \; h_1(x_1) \cdots h_n(x_n) 
		\nonumber \\  
		& = \lim_{\ell_i \to \infty} 
		\omega_\beta \bigl(  \phi_{\ell_1} (r_1, h_{1})   \cdots 
		\phi_{\ell_n} (r_n, h_{n})   \bigr)_{\upharpoonright r_i=t_i + i \alpha_i}   
	\end{align}
for $0 < \alpha_n < \ldots < \alpha_1 < \beta$. Now the H\"older inequality 
(\ref{crucialbound}) implies that each of the $2^n$ terms arising from the linear polar decomposition 
can be estimated:  for $0 < \alpha_n < \ldots < \alpha_1 < \beta$,  
$i= 1, \ldots, n$, we have
	\begin{align}
		\label{72new}
		\lim_{\ell_i \to \infty} & \left| 
		\omega_\beta \bigl(  \phi^{\pm}_{\ell_1} (r_1, h_{1})   \cdots 
		\phi^{\pm}_{\ell_n} (r_n, h_{n})   \bigr)_{\upharpoonright r_i=t_i + i \alpha_i} \right| 
		\nonumber \\  
		& \le 
		\lim_{\ell_i \to \infty}
		\| \phi^{\pm}_{\ell_1} (t_1, h_{1}) \|_{p_1} 
		\cdots \| \phi^{\pm}_{\ell_n} (t_n, h_{n}) \|_{p_n} \, 
		\nonumber \\  
		& \le   \frac{p_1 }{2} \cdots \frac{p_n }{2}
		 \cdot |h_1|_{\mathscr S}  \cdots |h_n|_{\mathscr S} \; ,
		 \qquad t_1, \ldots, t_n \in \mathbb{R} \; , 
	\end{align} 
with $ p_i \equiv p_i( \underline \alpha ) $
the smallest integer such that 
	\[
		\frac{1}{p_i(\underline \alpha)}  < \frac{1}{ \beta} \min \left\{ \alpha_{i+1}-  \alpha_{i} \; , \; \alpha_{i}-  \alpha_{i-1}\right\}
		\; ,  \qquad i=1, \ldots, n \; . 
	\]
(Setting $\alpha_0 =\beta - \alpha_n$ and $\alpha_{n+1} = \beta- \alpha_1$.) 
In the second inequality in~(\ref{72new}) 
we have used Lemma \ref{Lm7new} to conclude that for $p$ sufficiently 
large\footnote{Recall that   $p! < (p/2)^p$ for $p \ge 6$.} 
	\[
		\lim_{\ell \to \infty}
		\| \phi^{\pm}_{\ell} (t, h) \|_p 
		\le |\kern -.5mm |\kern -.5mm | \, h  |\kern -.5mm |\kern -.5mm |_{p }  
		\le  \sqrt[p]{p!} \cdot|h|_{\mathscr S} <  \frac{p}{2} \cdot|h|_{\mathscr S} \; . 
	\] 
Thus, for $0 < \alpha_n < \ldots < \alpha_1 < \beta$,  
	\begin{align*}
		\int & {\rm d}x_1 \cdots {\rm d}x_n \; {\mathbf W}
		({\underline s} - i {\underline \alpha}, {\underline y}) \; h_1(x_1) \cdots h_n(x_n) 
		\nonumber \\  
		& \qquad \le   \frac{p_1 ( \underline \alpha )}{2} \cdots \frac{p_n ( \underline \alpha )}{2}
		 \cdot |h_1|_{\mathscr S}  \cdots |h_n|_{\mathscr S} \; ,
		 \qquad t_1, \ldots, t_n \in \mathbb{R} \; . 
	\end{align*} 
Note that 
$ p_i(\lambda \underline \alpha) \sim \lambda^{-1} p_i(\underline \alpha) $ for $\lambda \searrow 0$.

We will now show, following ideas in \cite[p.~24]{RS}, that this bound ensures that the boundary values exist as
tempered distributions as $\alpha_i \searrow 0$: 
define, for $\lambda \in (0, 1]$ fixed,  a tempered distribution $T_{\underline \alpha}  
(\lambda) \in {\mathscr S}'(\mathbb{R}^{n-1})$ by setting, 
for $g \in {\mathscr S}(\mathbb{R}^{n-1})$, 
	\[
		T_{\underline \alpha}  (\lambda) (g)  
		:= \int_{\mathbb{R}^{n-1}} {\rm d} {\underline s} \; g   ({\underline s})
				\int  {\rm d}x_1 \cdots {\rm d}x_n \; {\mathbf W}
		({\underline s} - i \lambda {\underline \alpha}, {\underline y}) \; 
		h(x_1) \cdots h(x_n) 	 \;  .   
	\]
Let $T_{\underline \alpha}^{(k)}  (\lambda) $, $k=1, 2, \ldots$, denote the $k$-th distributional derivative, specified by setting
	\begin{align}
		\label{Tder}
		& T_{\underline \alpha}^{(k)}   (\lambda) (g) \\ 
		& \quad = \int_{\mathbb{R}^{n-1}} {\rm d} {\underline s} 
				\int  {\rm d}x_1 \cdots {\rm d}x_n \; {\mathbf W}
		({\underline s} - i \lambda {\underline \alpha}, {\underline y}) \; h_1(x_1) \cdots h_n(x_n) 
		\left( i {\underline \alpha} \cdot \frac{\partial}{ \partial 	{\underline s} }\right)^k 
		g ({\underline s}) \;.  \nonumber
	\end{align}
Thus, by the fundamental theorem of calculus,
	\begin{align}
	\label{pol0}
	T_{\underline \alpha}  (\lambda) 
		&= T_{\underline \alpha}  (1)   + \sum_{j=1}^{k-1} Q_j ( \lambda) \, T_{\underline \alpha}^{(j)}   (1) 
		\nonumber
		\\
		&  \qquad  - \int_{\lambda}^1 {\rm d} \lambda_k   
		\int_{\lambda_k}^1 {\rm d} \lambda_{k-1} \cdots 
		\int_{\lambda_{2}}^1 {\rm d} \lambda_1   \; 
		T_{\underline \alpha}^{(k)} ( \lambda_1)   \; .   
	\end{align}
The $Q_j$'s in (\ref{pol0}) are suitable polynomials. The limit $\lambda \downarrow 0$ in (\ref{pol0}) can be taken, provided that 
there exists a $k$ such that
	\begin{align}
		\label{pol}
		\lim_{ \lambda \downarrow 0} \left| 
		\int_{\lambda}^1 {\rm d} \lambda_k   
		\int_{\lambda_k}^1 {\rm d} \lambda_{k-1} \cdots 
		\int_{\lambda_{2}}^1 {\rm d} \lambda_1   \; 
		T_{\underline \alpha}^{(k)} ( \lambda_1) 
		(g) \right| < c \cdot \| g \|_{\mathscr S} \; ,   
	\end{align}
with $c>0$ a constant and  $\| g \|_{\mathscr S}$ a Schwartz semi-norm. 
This is done by estimating $T_{\underline \alpha}^{(k)}(\lambda) (g) $ as given in (\ref{Tder}) 
for $\lambda \in (0, 1]$:
choose some $ m \in \mathbb{N}$  large enough so that $\int_{\mathbb{R}^{n-1}} {\rm d}{\underline s}\, 
(1+ |{\underline s}|)^{-m} < \infty$. Then, for $\lambda \in (0, 1] $,
	\begin{align}
		\label{temp-bound}
		\left| T^{j}_{\underline \alpha} (\lambda) (g) \right|
		&= C \; \sup_{\underline{ t } \in \mathbb{R}^{n-1}}   
		| (1+ |{\underline s}|)^{m} | \; \left| \left( i {\underline \alpha} \cdot
		\frac{\partial}{ \partial {\underline s} }\right)^j g ({\underline s}) \right| 
			\nonumber \\
		&  \qquad \qquad \times \; 
		p_1 (\lambda  \underline{\alpha})\cdots p_{n} (\lambda  \underline{\alpha} ) \cdot |h_1|_{\mathscr S}   
		\cdots |h_{n}|_{\mathscr S}
		\nonumber \\
		&\le   C' \cdot\lambda^{-n} 		\; , \qquad C, C' >0 \;.   
	\end{align}
Note that 
	\begin{align}
		\label{pol1}
		\lim_{ \lambda \downarrow 0} \left| \int_{\lambda}^1 {\rm d} \lambda_k  \int_{\lambda_k}^1 {\rm d} \lambda_{k-1} \cdots 
		\int_{\lambda_{2}}^1 {\rm d} \lambda_1   \; \lambda_1^{-n}  \right| < \infty  \;     
	\end{align}
for $k$ sufficiently large, \emph{i.e.}, $k \ge n  + 1$. 
Combining (\ref{pol}), (\ref{temp-bound}), and (\ref{pol1}) one concludes  that 
the limit of~$T_{\underline \alpha} (\lambda)$ exists as $\lambda \downarrow 0$ and 
that each term in the limit is less than 
or equal to a constant times an ${\mathscr S}(\mathbb{R}^{n-1})$-seminorm of $g$. Thus 
	\[
		{\mathbf W} ({\underline s} - i   {\underline \alpha}, \underline{x}) 
		= 		
		{\mathcal W}_\beta^{(n-1)} 
		\bigl( s_1 - i (\alpha_1-\alpha_2), y_1, \ldots, 
		 s_{n-1} - i (\alpha_{n-1} -\alpha_{n}), y_{n-1} \bigr) \; . 
	\] 
converges in ${\mathscr S}'(\mathbb{R}^{n-1})$ as ${\underline \alpha} \downarrow 0$ to a 
tempered distribution. 
The latter is denoted by  ${\mathfrak W}_{\beta}^{(n-1)}(t_1 - t_2 , x_1 - x_2, \ldots, t_{n-1} - t_n , x_{n-1} - x_n)$. 
 
Now suppose that $e= ( \tau_1, z_1 , \ldots \tau_{n-1}, z_{n-1} ) \in \mathbb{R}^{2(n-1)}$ is 
an $(n-1)$-tuple of time-like unit vectors 
$(\tau_i, z_i)\in V^+$, and that $\widetilde h_i \in C^\infty_{0 \mathbb{R}}(\mathbb{R})$, $i=1,  \ldots, n$. 
To ease the notation we set 
	\begin{align*}
		\underline{h} ( \underline{x}) & = h_1 (x_1) \cdots h_n (x_n) \; , \\
	{\underline \alpha_\tau} & = \bigl(\tau_1(\alpha_1- \alpha_2), \ldots, \tau_{n-1}(\alpha_{n-1}- \alpha_n) \bigr) \; .
	\end{align*}
Then 
\begin{align*}
T_{\underline \alpha e} (\lambda) (g) & \doteq \int_{\mathbb{R}^{n-1}}  {\rm d} {\underline s} \; g ({\underline s})
				\int  {\rm d}x_1 \cdots {\rm d}x_n \; h_1 (x_1) \cdots h_n (x_n)
				\\
& \qquad \qquad \qquad \qquad \qquad \times {\mathbf W}
		\bigl( ({\underline s} , {\underline y})  - i \lambda {\underline \alpha} e 
		+ i \lambda {\underline \alpha} (e-  (1, 0) ) \bigr) 
		\\
& =  \int_{\mathbb{R}^{n-1}} {\rm d} {\underline s} \; 
g \bigl({\underline s} \bigr)		
\int  {\rm d}x_1 \cdots {\rm d}x_n \; h(x_1 + i \lambda z_1) \cdots h(x_{n-1} + i \lambda z_{n-1} )h(x_n)
				\\
& \qquad \qquad \qquad \qquad \qquad \times {\mathbf W}
		\bigl( ({\underline s} - i \lambda {\underline \alpha_\tau}, {\underline y})   \bigr) 
% \\
% & =  
% T_{\underline \alpha, e} (\lambda) 
% \bigl( \widetilde{{\rm e}^{- \lambda (e-  (1, 0) ) 
% \underline{\alpha} \cdot \underline{k} } \; \widetilde {(g \otimes \underline{h} )} (\underline{k} )} \bigr) 
\; , 
\end{align*}
where we have used the fact that the $h_i $'s, $i=1,  \ldots, n$, are entire 
and the estimates in the Paley-Wiener theorem (Theorem IX.11 \cite{RS}) to shift the hyperplane of integration 
in second equality. 

Since $\widetilde {\underline{h}}_i \in C^\infty_{0 \mathbb{R}}(\mathbb{R})$, 
$h(x_i + i \lambda z_i) \to h (x_i)$ as $\lambda \searrow 0$. 
Since such $h_i $'s, are dense in ${\mathscr S}( \mathbb{R})$,   
	\[
		\lim_{ \lambda \downarrow 0} T_{\underline \alpha, e} (\lambda) (g)  = 
		T_{{\underline \alpha}_\tau} (0) \bigl( g \bigr)
		= T_{\underline \alpha} (0) \bigl( g \bigr) \; .  
	\]
Thus, the limit of 
${\mathbf W} \bigl( ({\underline s},\underline{x}) - i   {\underline \alpha} e \bigr) $ coincides 
with the tempered distribution
	\[
		{\mathfrak W}_{\beta}^{(n-1)}(t_1 - t_2 , x_1 - x_2, \ldots, t_{n-1} - t_n , x_{n-1} - x_n)
	\]
encountered before. 

The KMS boundary condition
follows by differentiating (see (\ref{stone})) the boundary condition of the 
corresponding Weyl operators given 
in (\ref{KMSWEYL}).
% \qed
\end{proof}

\begin{remark} We note that the thermal Wightman distributions
	\[
		{\mathfrak W}_{\beta}^{(n-1)}(t_1 - t_2 , x_1 - x_2, \ldots, t_{n-1} - t_n , x_{n-1} - x_n)
	\]
are analytic functions as long as the $(t_i, x_i)$, $i =1, \ldots n$, are mutually space-like points. This can be shown by an 
argument similar to the one outlined in the discussion preceding Theorem \ref{Th5}.
\end{remark}

\section{Summary and Outlook}

For quite some time the pioneering work of H\o egh-Krohn \cite{H-K} did not find the recognition it deserves. 
However, the thermal ${\mathscr P}(\varphi)_2$ model 
should be seen as a binding link between statistical mechanics and quantum field theory. 
The authors believe that by providing more detail on the construction of this model (see \cite{GeJ1}\cite{GeJ2})
and verifying that it satisfies key axioms, other scientists might get motivated to look at this model in more detail. 
In fact, the physical properties of this model have hardly been explored so far. It would be interesting to know how, for instance, 
the specific heat behaves as a function of the temperature and the coupling constants. A more challenging 
question is to investigate the particle content of this model. Eventually, one may want to set up scattering theory at
positive temperature or prove the uniqueness of the KMS state for all temperatures and all allowed values of the 
coupling constant. There are strong indications 
that the correlation functions decay exponentially in space-like directions, and thus it seems to the authors 
that all of these questions can be resolved with reasonable amount of work. 

\paragraph{Acknowledgement}

This work was supported by the Leverhulme Trust [F/00 407/BC -- Particles and Thermal 
Quantum Fields]. 
We would like to thank the second referee for suggesting several improvements and 
requesting a number of clarifications. We also would like to thank 
Christopher J. Fewster for pointing out a gap in the proof of Theorem 5 in an earlier version of 
the manuscript.

\end{document}